\documentclass[11pt, letterpaper]{amsart}
\usepackage{graphicx, amssymb, hyperref}
\usepackage{amsmath}
\usepackage[round]{natbib}
\usepackage{color}

\newtheorem{thm}{Theorem}[section]

\newtheorem{lem}[thm]{Lemma}
\newtheorem{prop}[thm]{Proposition}

\theoremstyle{definition}
\newtheorem{defn}[thm]{Definition}

\newtheorem{ass}[thm]{Assumption}

\theoremstyle{remark}
\newtheorem{rem}[thm]{Remark}
\newtheorem{exa}[thm]{Example}
\numberwithin{equation}{section}

\newcommand{\abs}[1]{\left\vert#1\right\vert}
\newcommand{\set}[1]{\left\{#1\right\}}
\newcommand{\Real}{\mathbb R}

\newcommand{\Natural}{\mathbb N}
\newcommand{\nin}{n \in \Natural}
\newcommand{\kin}{k \in \Natural}

\newcommand{\such}{{\ | \ }}
\newcommand{\limn}{\lim_{n \to \infty}}
\newcommand{\limsupn}{\limsup_{n \to \infty}}
\newcommand{\liminfn}{\liminf_{n \to \infty}}

\newcommand{\dfn}{\, := \,}

\newcommand{\Pre}{\mathcal{P}}
\newcommand{\prob}{\mathbb{P}}

\newcommand{\plim}{\prob \text{-} \lim}
\newcommand{\plimn}{\plim_{n \to \infty}}
\newcommand{\pliminf}{\prob \text{-} \liminf}
\newcommand{\pliminfn}{\pliminf_{n \to \infty}}

\newcommand{\expec}{\mathbb{E}}

\newcommand{\basis}{(\Omega, \mathcal{F}, \mathbf{F}, \prob)}

\newcommand{\F}{\mathcal{F}}

\newcommand{\cadlag}{c\`adl\`ag}

\newcommand{\ud}{\mathrm d}

\newcommand{\zo}{[0, 1)}

\newcommand{\num}{num\'eraire}
\newcommand{\X}{\mathcal{X}}

\newcommand{\Xo}{\X^{\circ}}

\newcommand{\rr}{\mathsf{rr}}
\newcommand{\rrt}{\rr_\tau}
\newcommand{\rri}{\rr_\infty}
\newcommand{\rrh}{\rr_\hor}

\newcommand{\err}{\expec \rr}
\newcommand{\errt}{\err_\tau}
\newcommand{\errh}{\err_\hor}
\newcommand{\erri}{\err_\infty}

\newcommand{\Shn}{\mathcal{S}_{\hor} \text{-} \limn}

\newcommand{\Stkn}{\mathcal{S}_{\tau_k} \text{-} \limn}
\newcommand{\Sln}{\mathcal{S}_{\mathsf{loc}} \text{-} \limn}
\newcommand{\Stln}{\mathcal{S}_{\tau_\ell} \text{-} \limn}

\newcommand{\Xhat}{\widehat{X}}

\newcommand{\pare}[1]{\left(#1\right)}
\newcommand{\bra}[1]{\left[#1\right]}
\newcommand{\dbra}[1]{[#1]}
\newcommand{\dbraco}[1]{[\kern-0.15em[ #1 [\kern-0.15em[}
\newcommand{\dbraoc}[1]{]\kern-0.15em] #1 ]\kern-0.15em]}

\newcommand{\bF}{\mathbf{F}}

\newcommand{\indic}{\mathbb{I}}

\newcommand{\absco}{{<\kern-0.53em<}}

\newcommand{\xhat}{{\widehat{X}}}
\newcommand{\xt}{{\widetilde{X}}}

\newcommand{\zt}{{\widetilde{Z}}}

\newcommand{\hha}{{}^{\alpha} \kern-0.23em {\widehat{H}}}
\newcommand{\hta}{{}^{\alpha} \kern-0.23em {\widetilde{H}}}

\newcommand{\tX}{{}^{\tau} \kern-0.23em \X}
\newcommand{\tx}{{}^{\tau} \kern-0.23em X}

\newcommand{\Xa}{{}^{\alpha} \kern-0.23em \X}
\newcommand{\Xb}{{}^{\beta} \kern-0.23em \X}
\newcommand{\Xz}{{}^{0} \kern-0.23em \X}

\newcommand{\xa}{{}^{\alpha} \kern-0.23em X}
\newcommand{\xha}{{}^{\alpha} \kern-0.23em \xhat}
\newcommand{\xta}{{}^{\alpha} \kern-0.23em \xt}

\newcommand{\xb}{{}^{\beta} \kern-0.23em X}
\newcommand{\xhb}{{}^{\beta} \kern-0.23em \xhat}

\newcommand{\rhoh}{\widehat{\rho}}
\newcommand{\Dh}{\widehat{D}}

\newcommand{\el}{\exp(\ell)}

\newcommand{\ela}{\exp((1 - \alpha) \ell)}

\newcommand{\apix}{{}^{\alpha} \kern-0.17em \pi^X}

\newcommand{\hor}{T}

\newcommand{\Xoo}{\X^{\circ \circ}}
\newcommand{\chic}{\check{\chi}}
\newcommand{\Time}{\mathbb{T}}
\newcommand{\Yhat}{\widehat{Y}}

\begin{document}

\title[Num\'eraire Property and Long-Term Growth with Drawdown Constraints]{The Num\'eraire Property and Long-Term Growth Optimality for Drawdown-Constrained Investments}%

\author{Constantinos Kardaras}%
\address{Constantinos Kardaras, Statistics Department, London School of Economics and Political Science, 10 Houghton Street, London, WC2A 2AE, UK.}%
\email{kardaras@bu.edu}%

\author{Jan Ob{\l}\'oj}%
\address{Jan Ob{\l}\'oj. Mathematical Institute and
Oxford-Man Institute of Quantitative Finance, University of Oxford, Oxford OX1 3LB, UK.}%
\email{obloj@maths.ox.ac.uk}%

\author{Eckhard Platen}%
\address{Eckhard Platen. Finance Discipline Group \& School of Mathematical Sciences,
University of Technology, Sydney, P.O. Box 123, Broadway, NSW 2007, Australia.}%
\email{eckhard.platen@uts.edu.au}%

\thanks{The first author acknowledges partial support by the National Science Foundation, grant number DMS-0908461. The second author gratefully acknowledges the support of the 2011 Bruti-Liberati Fellowship at University of Technology, Sydney. The authors are grateful for helpful comments to the participants of the SIAM Annual Meeting in Minneapolis, workshop on Mathematical Finance in Kyoto, Optimal Stopping Workshop in Warwick as well as seminar participants at the Oxford-Man Institute for Quantitative Finance.
}

\date{\today}%
\begin{abstract}
We consider the portfolio choice problem for a long-run investor in a general continuous semimartingale model. We suggest to use path-wise growth optimality as the decision criterion and encode preferences through restrictions on the class of admissible wealth processes. Specifically, the investor is only interested in strategies which satisfy a given linear drawdown constraint. 
 The paper introduces the {\num} property through the notion of expected relative return and shows that drawdown-constrained strategies with the \num \ property exist and are unique, but may depend on the financial planning horizon. However, when sampled at the times of its maximum and asymptotically as the time-horizon becomes distant, the drawdown-constrained \num\ portfolio is given explicitly through a model-independent transformation of the unconstrained \num\ portfolio. Further, it is established that the asymptotically growth-optimal strategy is obtained as limit of {\num} strategies on finite horizons.
\end{abstract}

\maketitle


\section*{Introduction}

\subsection*{Discussion}

The debate whether an investor with a long financial planning horizon should use the growth-optimal strategy, as postulated by \cite{KELLY}, is among the oldest in the portfolio selection literature, see \cite{MacLean2010}. In particular, opposite sides were assumed by two among the most prominent scholars in the field: while Paul Samuelson fiercely criticised the use of Kelly's strategy, including the famous refute \cite{Samuelson:79} in words of one-syllable, Harry Markowitz argued for it already in his 1959 book, see also \cite{Markowitz:06}. The arguments in favour of Kelly's investment strategy rely on the fact that asymptotic growth should be of prime interest for long-run investment. More recently, this line of argument has seen a revived interest in particular through the so-called benchmark approach, see \cite{MR2267213}. The arguments against it point to the fact that the growth-rate maximisation does not take into account investor's risk appetite and is too simplistic. Samuelson, as well as many others including seminal works of \cite{Merton:71}, looked instead at maximising expected utility.
While Kelly's strategy itself falls into this category, with the utility function being the logarithmic one, choices of other utility functions result in criteria that can accommodate different risk preference profiles. 

Expected utility maximisation (EUM), with its axiomatic foundation going back to \cite{MR0021298} and \cite{Savage:1954wo}, is probably the most successful and widely studied framework for normative decisions under uncertainty. 
However, although quite flexible, this line of reasoning is vulnerable to critique as it involves several often arbitrary choices. For example, it is hard to justify why the investors should think in terms of, and be able to specify, their utility functions. Utility elucidation methods systematically yield different results, see e.g.\ \cite{Hershey:1985bn}, and answers incompatible with the EUM paradigm, see e.g.\ \cite{Kahneman:1979wl}. This led to many ramifications, including the behavioural portfolio selection of \cite{Kahneman:1979wl}, see also \cite{Jin:2008ek}.
Further, optimal investment decisions of investors maximising their expected utility typically depend on the choice of time-horizon, which is a rather arbitrary input, particularly so for a long-run investor.\footnote{For an attempt to circumvent the horizon-dependence, see recent developments in \cite{MusielaZariphopoulou:09}.} Finally, the resulting optimal investment strategies entangle in a rather complex way the choice of the model for the dynamics of the risky assets and the choice of the investors' preferences (utility function). Attempts to account for \emph{Knightian uncertainty} resulted in a ``robust" version of the EUM, see e.g.\ \cite{Gilboa:1989hm}, \cite{MaccheroniMarinacciRustichini:06}, \cite{Ju:2012kf}.

In this work, we explore a potential way to bridge the two opposing sides and obtain a decision mechanism which is based on the appealingly simple Kelly principle of dominating any other investment in the long run while at the same time incorporating a flexible specification of risk preferences. Furthermore, the resulting optimal strategies effectively disentangle the effects of model and preferences  choices. Instead of specifying the risk attitudes through a utility function, we propose to encode them as a restriction on the universe of acceptable investment strategies (or, equivalently, acceptable wealth evolutions). More specifically, we impose a drawdown constraint and only consider portfolios which never fall below a given fraction $\alpha$ of their past maximum.\footnote{If $(X_t)$ is the investor's wealth process, expressed in units of some baseline asset, then we require $X_t\geq \alpha \sup_{u \in [0, t]} X_u$ for all $t$ up to the considered, possibly random or infinite, time horizon $T$.} Such interpretation of ``attitude towards risk'' is in fact commonly utilised in practice\footnote{Performance measures involving drawdowns include e.g.\ Calmar ratio, Sterling ratio and Burke ratio, see \cite{Eling:2007ut} and Chapter~4 of \cite{Bacon:2008um}. For a further discussion on their practical use see \cite{Lhabitant:2004wa}. Drawdowns are often reported, e.g.\ the Commodity Futures Trading Commission's mandatory disclosure regime stipulates that managed futures advisers report their ``worst peak-to-valley drawdown."}: we effectively equip the investor with a stop-loss safety trigger to avoid large drawdowns. The constraint implies that even the most adverse market crash may not reduce wealth by more than $100(1-\alpha)\%$.

Using drawdown constraints to encode risk attitudes is not only practically motivated but has also sound theoretical basis. The parameter $\alpha\in [0,1)$ allows for flexible specification of risk preferences and encapsulates the risk aversion.
Drawdown constraints were first considered in a continuous-time framework by \cite{GZ}, then by \cite{CvitanicKaratzas:94} and more recently by \cite{ChernyObloj:11}. These contributions focused on maximising the growth rate of expected utility and show that imposing drawdown constraints is essentially equivalent to changing investor's risk aversion. More precisely, \cite{ChernyObloj:11} consider two investors in a general semimartingale model: one endowed with a power (HARA) utility with risk aversion $\gamma$ and facing an $\alpha$-drawdown constraint and another with risk aversion $(\gamma + \alpha(1-\gamma))$ and no constraints. They prove that the two are equivalent in the sense that they both achieve the same asymptotic growth rate of expected utility, and that their optimal portfolios are related through an explicit model-independent transformation. 

Having encoded preferences through a drawdown constraint, we consider decision making based on optimality expressed through the \emph{\num \ property} in the spirit of \cite{Long:1990hm}. We require that expected relative returns of any other non-negative investment with respect to the optimal portfolio over the same time-period are non-positive. In fact, this choice of optimality arises in an axiomatic way from \num\footnote{Here in the sense of \emph{a} \num, e.g.\ currency.}-invariant preferences, as set forth in \cite{MR2724418}. 
In the unconstrained case, the global {\num} portfolio $\Xhat$ is a wealth process which has the property that all other investments, denominated in units of $\Xhat$, are supermartingales. It is well known that $\Xhat$ also maximises the asymptotic long-term growth-rate and is exactly the investment corresponding to Kelly's criterion, see e.g.\ \cite{Hakansson:1971kp,BansalLehmann:97} and the references therein.
Some recent contributions explored the {\num} property under a constrained investment universe. In particular, \cite{MR2335830} showed that with point-wise convex constraints on the proportions invested in each asset, one can retrieve existence and all useful properties of the \num \ portfolio. We contribute to this direction of research by providing a detailed analysis of the {\num} property within the class of investments which satisfy a given linear drawdown constraint, where wealth can never fall below a fraction $\alpha$ of its running maximum. 

Our first main result establishes existence of unique portfolios with the \num \ property over different time-horizons under drawdown-constrained investment. However, in contrast to the unconstrained case, the optimal strategies may depend on the time horizon and are no longer myopic. Our results are valid in a general continuous-path semimartingale set-up. 
We also discuss detailed structural asymptotic properties of the optimal strategies, including a version of the so-called turnpike theorem. 

Our second main result considers a long-run investor. Given the investor's acceptable level of drawdown $\alpha$, we show that there is a unique choice of investment strategy which almost surely asymptotically outperforms any other strategy which satisfies the $\alpha$-drawdown constraint. In this way, we succeed in using the Kelly criterion while allowing for a flexible specification of risk attitudes. 
The optimal strategy is given explicitly in two manners. First, we obtain a version of the mutual fund theorem: the optimal strategy is a dynamic version of the so-called fractional Kelly strategy. It invests a fraction of wealth, which depends on the current level of drawdown, in the fund represented by $\Xhat$ and the remaining fraction in the baseline asset. Second, the optimal strategy is given as a pathwise and model-independent transformation of the unconstrained \num\ strategy $\Xhat$. As a result, it disentangles the effects of model specification and preferences specification. The former yields the Kelly strategy $\Xhat$. The latter specifies the transformation which is applied to $\Xhat$ to control the risk by avoiding drawdowns beyond a certain magnitude. In this way, long-run investment decisions are decomposed in two separate steps. The modelling effort is reduced to constructing the global \num\ portfolio $\Xhat$, or the best approximation thereof. Preferences elucidation is reduced to determining the acceptable level of drawdown parameter $\alpha\in [0,1)$.

An important tool in our study is the so-called Az\'ema--Yor transformation, a result in stochastic analysis which allows to build an explicit, model--independent, bijection between all wealth processes and wealth processes satisfying a drawdown constraint. This transformation was established in a general semimartingale set-up in \cite{BKO} and used by \cite{ChernyObloj:11} in a utility maximisation setting. However, a special case of it was already used in \cite{CvitanicKaratzas:94}. We show here that the Az\'ema--Yor transform $\xha$ of $\Xhat$ has the {\num} property within the class of portfolios satisfying the $\alpha$-drawdown constraint, both in an asymptotic sense and when sampled at times of its maximum. However, the optimal strategies may depend on the time horizon; in particular, it is not true that all other drawdown-constrained wealths are supermartingales in units on $\xha$, a feature often used previously to define the \num\ property, see \cite{Long:1990hm} or \cite{MR2267213}. Finally, we show that $\xha$ is the only wealth process which enjoys the {\num} property along some increasing sequence of stopping times that tends to infinity. Further, portfolios enjoying the \num \ property for investment with long time-horizons are close (in a very strong sense) to $\xha$ at initial times.
The optimal portfolio $\xha$, as mentioned above, can be explicitly produced by investing at each time a fraction of current wealth in the fund represented by $\Xhat$ and the remaining fraction in the baseline asset.
When the domestic savings account is taken as the baseline asset, $\Xhat$ and $\xha$ have the same instantaneous Sharpe ratio. Both portfolios are located at the Markowitz efficient frontier. However $\xha$ trades off long-term growth for a path-wise capital guarantee in the form of a drawdown constraint, something $\Xhat$, or solutions to expected utility maximisation in general, can not offer.

We stress the fact that the results presented here do not follow from previous literature because of the generality of our setup and the complex nature of the drawdown constraints. In fact, novel characteristics appear in this setting. For example, portfolios with the \num \ property are no longer myopic, and depend on the financial planning horizon. Interestingly, it is one of our main points that in an asymptotic sense the myopic structure is reinstated. Furthermore, we emphasise that the findings of this paper are essentially model-independent and, therefore, rather robust. Finally, we wish to draw some attention to the underlying philosophy relative to the practical perspective. A long-run investor will only witness a single realisation of the market dynamics. Therefore, path-wise outperformance is a very appealing decision criterion. We show that it is possible to combine it with risk preferences, and that this is best done not by complicating the investor's decision criteria, but rather by restricting the universe of acceptable trajectories of wealth evolution.

As already outlined above, drawdown constraints have features appealing to various participants in financial markets and are, thus, often encountered in practice, in either explicit or implicit manner. For an investor in a fund, past performance often serves as a benchmark. A large drawdown---the relative difference between the past maximum and the current value---would indicate and be perceived as a worsening of performance and even a loss. Accordingly, it may be used as a stop-loss trigger. This, in turn, would usually lead to a flight of capital from the fund, a threatening situation that should be avoided from a managerial perspective. Note that drawdown constraints may also result implicitly from the structure of hedge fund manager's incentives through the high-water mark provision---see e.g.\ \cite{GuasoniObloj:11}.

Despite their practical importance, there are relatively few theoretical studies on portfolio selection with drawdown constraints. The main obstacle is the inherent difficulty associated with such non-myopic constraint: it involves the running maximum process and, therefore, depends on the whole history of the process. Apart from the contributions mentioned above, we note that \cite{Magdon-IsmailAtiya:04} derived results linking the maximum drawdown to average returns. In \cite{CUZ:05}, the problem of maximising expected return subject to a risk constraint expressed in terms of the drawdown was considered and solved numerically in a simple discrete time setting.
Finally, in continuous-time models, drawdown constraints were also recently incorporated into problems of maximising expected utility from consumption---see \cite{Elie:08} and \cite{ET}. Options on drawdowns were also explored as instruments to hedge against portfolio losses, see \cite{Vecer:06}. Further, the maximisation of growth subject to constraints arising from alternative risk measures is discussed in \cite{MR2536869}.

\subsection*{Structure of the paper}
Section \ref{sec: market} contains a description of the financial market and introduces drawdown-constrained investments. In Section \ref{sec: num_property}, the \num \ property of drawdown-constrained investments is explored. Main results are Theorem \ref{thm: exist}, establishing existence and uniqueness of portfolios with the \num \ property for finite time-horizons, and Theorem \ref{thm: num}, which explicitly describes an investment that has the \num \ property at special stopping times where it achieves its maximum---in particular, this includes its asymptotic \num \ property. More asymptotic optimality properties of the aforementioned investment are explored in Section \ref{sec: more_asympt_optimal}. More precisely, its asymptotic (or long-run) growth-optimality is taken up in Theorem \ref{thm: asympt_growth}, and an important result in the spirit of turnpike theorems is given in Theorem \ref{thm: turnpike}. All the proofs are collected in Appendix \ref{sec: proofs}. Finally, in Appendix \ref{sec: careful_turnpike} we present an example in order to shed more light on the conclusion of the turnpike-type Theorem \ref{thm: turnpike}.

\section{Market and Drawdown Constraints}\label{sec: market}

\subsection{Financial market} \label{subsec: model}
We consider a very general frictionless financial market model with the assumption of \emph{continuous} price processes. Specifically, on a stochastic basis $\basis$, where $\bF = (\F_t)_{t \in \Real_+}$ is a filtration satisfying the usual hypotheses of right-continuity and saturation by $\prob$-null sets of $\F$, let $S = (S^1, \ldots, S^d)$ be a $d$-dimensional semimartingale with a.s.\ continuous paths---see, for example, \cite{MR1121940}. Each $S^i$, $i \in \set{1, \ldots, d}$, is modelling the random movement of an asset price in the market, discounted by a baseline asset. It is customary to assume that the baseline asset is the (domestic) savings account, but it does not necessarily have to be so. 

Define $\X$ to be the class of all \emph{nonnegative} processes $X$ of the form
\begin{equation} \label{eq: weath}
X = 1 + \int_0^\cdot \pare{H_t, \ud S_t} \equiv 1 + \int_0^\cdot \pare{\sum_{i=1}^d H^i_t \ud S^i_t},
\end{equation}
where $H = (H^1, \ldots, H^d)$ is a $d$-dimensional predictable and $S$-integrable\footnote{All integrals are understood in the sense of vector stochastic integration.This somewhat mathematical restriction has proved essential in order to formulate elegant versions of the Fundamental Theorem of Asset Pricing, as well as to ensure that optimal wealth processes exist---for example, it is crucial for the validity of Theorems \ref{thm: asss_num} and \ref{thm: asss}, which are used extensively throughout the paper.} process. Throughout the paper, $(\cdot, \cdot)$ is used to (sometimes, formally) denote the inner product in $\Real^d$.   The process $X$ of \eqref{eq: weath} represents the outcome of trading according to the investment strategy $H$, denominated in units of the baseline asset.  In the sequel we are interested in ratios of portfolios; therefore, the initial value $X_0$ plays no role as long as it is the same for all investment strategies. For convenience we assume $X_0 = 1$ holds for all $X \in \X$.

In the following, we characterise in a precise manner the rich world of models that we permit for our market. These include most continuous-path models that have been studied in the literature. Essential is the existence of the \emph{\num \ portfolio}---see \cite{Long:1990hm}. However, existence of an equivalent risk-neutral probability measure is not requested; therefore, certain forms of classical arbitrage are permitted.

\begin{defn}
We shall say that there are opportunities for arbitrage of the first kind if there exist $T \in \Real_+$ and an $\F_T$-measurable random variable $\xi$ such that:
\begin{itemize}
	\item $\prob[\xi \geq 0] = 1$ and $\prob[\xi > 0] > 0$;
	\item for all $x > 0$ there exists $X \in \X$, which may depend on $x$, with $\prob[x X_T \geq \xi] = 1$.
\end{itemize}
\end{defn}

The following mild and natural assumption is key to the development of the paper.

\begin{ass} \label{ass: basic}
In the market described above, the following hold:
\begin{enumerate}
	\item[(A1)] There is no opportunity for arbitrage of the first kind.
	\item[(A2)] There exists $X \in \X$ such that $\prob \bra{\lim_{t \to \infty} X_t = \infty} = 1$.
\end{enumerate}
\end{ass}

Condition (A1) in Assumption \ref{ass: basic} is a minimal market viability assumption. On the other hand, condition (A2) asks for sufficient market growth in the long run. They are equivalent to the existence and growth condition of the \num\ portfolio.

\begin{thm} \label{thm: asss_num}

Condition (A1) of Assumption \ref{ass: basic} is equivalent to:
\begin{enumerate}
	\item[(B1)] There exists $\Xhat  \in \X$ such that $X/ \Xhat$ is a (nonnegative) local martingale for all $X \in \X$.
\end{enumerate}
Under the validity of (A1) or (B1), condition (A2) of Assumption \ref{ass: basic} is equivalent to:
\begin{enumerate}
	\item[(B2)] $\prob \big[ \lim_{t \to \infty} \Xhat_t = \infty \big] = 1$.
\end{enumerate}
\end{thm}

\begin{rem} \label{rem: num_log}
The equivalence of (A1) and (B1) was first discussed in \cite{Long:1990hm}. 
If the process $\Xhat$ in (B1) exists, then it is unique and is said to have the \textsl{\num \ property}. It is well known that it solves the log-utility maximisation problem on any finite time horizon, and that it achieves optimal asymptotic (or long-term) growth. We shall revisit these properties in a more general setting---see Remark \ref{rem: num implies log} and Theorem \ref{thm: asympt_growth}.
\end{rem}

The proof of Theorem \ref{thm: asss_num} is given in Subsection \ref{subsec: equivalentass} of Appendix \ref{sec: proofs}. In fact, it is a special case of a more general Theorem \ref{thm: asss} therein which contains several useful equivalent conditions to the ones presented in Assumption \ref{ass: basic}. 

\subsection{Drawdown constraints}
To each wealth process $X \in \X$, we associate its \textsl{running maximum} process $X^*$ defined via $X^*_t \dfn \sup_{u \in [0, t]} X_u$ for $t \in \Real_+$. The difference $X^* -X$ between the current wealth and its running maximum is called the \textsl{drawdown process}. As we argued in the introduction, different participants in financial markets may be interested to restrict the universe of their strategies to the ones which do not permit for drawdowns beyond a fixed fraction of the wealth's running maximum.

For any $\alpha \in \zo$, we write $\Xa$ for the class of wealth processes $X \in \X$ such that $X^*_t -X_t \leq (1-\alpha)X^*_t$, for all $t\geq 0$. The $[0,1]$-valued process $X / X^*$ is called the \textsl{relative drawdown} process associated to $X$.  Note that $X\in \Xa$ if and only if $X / X^* \geq \alpha$ holds identically.
It is clear that $\Xb \subseteq \Xa$ for $0 \leq \alpha \leq \beta < 1$, and that $\Xz = \X$. Note that if $X \in \X$ satisfies $X \geq \alpha X^*$ on the interval $\dbra{0, T}$ (here, $T$ can be any stopping time), then $(X_{T \wedge t})_{t \in \Real_+} \in \Xa$; therefore, it is appropriate to use $\Xa$ as the set of wealth processes regardless of the investment horizon. 

Interestingly, there is a one-to-one correspondence between wealth processes in $\X$ and wealth processes in $\Xa$ for any $\alpha \in \zo$.
The bijection was derived explicitly in terms of the so--called Az\'ema--Yor processes in \cite{BKO}, Theorem 3.4, and recently exploited in \cite{ChernyObloj:11}, in a general setting of possibly non-linear drawdown constraints. This elegant machinery simplifies greatly in the case of ``linear'' drawdown constraints considered here, and we provide explicit arguments, similarly to the pioneering work of \cite{CvitanicKaratzas:94}.
We first discuss how processes in $\X$ generate processes in $\Xa$---the converse will be established in the proof of Proposition \ref{prop: dd_charact} in the Appendix. For $X \in \X$ and $\alpha \in \zo$, define a process $\xa$ via\footnote{In the notation of \cite{BKO}, we have $\xa=M^{F_\alpha}(X)$ with $F_\alpha : \Real_+ \mapsto \Real_+$ defined via $F_\alpha(x)=x^{1-\alpha}$ for $x \in \Real_+$. Furthermore, Proposition 2.2 therein implies that $X=M^{G_\alpha}(\xa)$ with $G_\alpha=F_\alpha^{-1}$. This last converse construction is presented explicitly in Proposition \ref{prop: dd_charact}.}
\begin{equation} \label{eq: xa_defn}
\xa \dfn \alpha (X^*)^{1 - \alpha} + (1 - \alpha) X (X^*)^{- \alpha}.
\end{equation}
Using the fact that $\int_{0}^\infty \indic_{\set{X_t < X^*_t}} \ud X^*_t = 0$ a.s. holds, an application of It\^o's formula gives
\begin{equation} \label{eq: xa_SDE}
\xa = 1 + \int_0^\cdot (1 - \alpha) (X^*_t)^{- \alpha} \ud X_t,
\end{equation}
which implies that $\xa \in \X$. Furthermore, \eqref{eq: xa_defn} gives $\alpha (X^*)^{1 - \alpha} \leq \xa \leq (X^*)^{1 - \alpha}$. Note also that times of maximum of $X$ coincide with times of maximum of $\xa$ and consequently $\xa^* = (X^*)^{1 - \alpha}$. It follows that
\begin{equation} \label{eq: dd_xa}
\frac{\xa}{\xa^*} = \frac{\alpha (X^*)^{1 - \alpha} + (1 - \alpha) X (X^*)^{- \alpha}}{(X^*)^{1 - \alpha}} = \alpha + (1 - \alpha) \frac{X}{X^*}\geq \alpha,
\end{equation}
implying $\xa \in \Xa$.

The converse construction is presented in Proposition \ref{prop: dd_charact} below, the proof of which is given in Subsection \ref{subsec: proof of prop: dd_charact} of Appendix \ref{sec: proofs}. Together with \eqref{eq: xa_defn} they provide an extremely convenient representation of the class $\Xa$ for $\alpha \in \zo$, which we use extensively throughout the paper.

\begin{prop}[Proposition 2.2 of \cite{BKO}] \label{prop: dd_charact}
It holds that $\Xa = \set{\xa \such X \in \X}$.
\end{prop}


\begin{rem}\label{rem: SDEintepret}
One can rewrite equation \eqref{eq: xa_SDE} in differential terms as
\[
\frac{\ud \xa_t}{\xa_t} = \pare{\frac{(1 - \alpha) (X^*_t)^{- \alpha} X_t}{\xa_t}} \frac{\ud X_t}{X_t} =
\frac{\xa_t-\alpha\xa_t^*}{\xa_t}
\frac{\ud X_t}{X_t},
\]
for $t < \inf\{u \in \Real_+ \such X_u=0 \}=\inf\{u\in \Real_+ \such \xa_u-\alpha\xa_u^*=0\}$.
The above equation carries an important message: for $X \in \X$, the way that $\xa$ is built is via investing a proportion
\[
\apix \dfn \frac{\xa-\alpha\xa^*}{\xa} =  1 - \frac{\alpha}{\xa / \xa^*} = \frac{(1 - \alpha) (X / X^*)}{\alpha + (1 - \alpha) (X / X^*)}
\]
in the fund represented by $X$, and the remaining proportion $1 - \apix$
in the baseline asset. In particular, when the baseline asset is the domestic savings account, it follows that the Sharpe ratios of $X$ and $\xa$ are the same. Note that $0 \leq \apix \leq 1 - \alpha$ (so that $\alpha \leq 1 - \apix \leq 1$). Furthermore, $\apix$ depends only on $\alpha \in \zo$ and the relative drawdown $X /X^*$ of $X$. In fact, the proportion $\apix$ invested in the underlying fund represented by $X$ is an increasing function of the relative drawdown $X/X^*$.

Recall the \num \ portfolio process $\Xhat$ in (B1) in Theorem \ref{thm: asss_num}. When the above discussion is applied to $\xha$, defined from $\Xhat$ via \eqref{eq: xa_defn},
it follows from \cite[Theorem 11.1.3 and Corollary 11.1.4]{MR2267213} that $\xha$ is a locally optimal portfolio, in the sense that it locally maximises the excess return over all investments with the same volatility. In view of \eqref{eq: num} in Appendix \ref{sec: proofs}, the wealth process $\Xhat$ is given explicitly in terms of the drift and quadratic covariation process of the multi-dimensional asset-price process. It follows that $\xha$ for $\alpha \in \zo$ is explicitly specified as well.
\end{rem}

Even though the \num \ portfolio $\Xhat$ has optimal growth in an asymptotic sense (in this respect, see also Theorem \ref{thm: asympt_growth} later in the text), it is a quite risky investment. In fact, it experiences arbitrarily large flights of capital, as its relative drawdown process $\Xhat / \Xhat^*$ will become arbitrarily close to zero infinitely often. This is in fact equivalent to the following, seemingly more general statement, showing an oscillatory behavior of the relative drawdown for all wealth processes $\xha$, $\alpha \in \zo$.

\begin{prop} \label{prop: dd_osc}
Under Assumption \ref{ass: basic}, it holds that
\[
\alpha = \liminf_{t \to \infty} \pare{\frac{\xha_t}{\xha^*_t}} < \limsup_{t \to \infty} \pare{\frac{\xha_t}{\xha^*_t}} = 1,\ \textrm{a.s.} \quad \forall \alpha\in [0,1).
\]
\end{prop}

The proof of Proposition \ref{prop: dd_osc} is given in Subsection \ref{subsec: proof of prop: dd_osc} of Appendix \ref{sec: proofs}.

\section{The Num\'eraire Property} \label{sec: num_property}

\subsection{Expected relative return}

Fix a stopping time $\hor$ and $X,X' \in \X$, and define the \textsl{return of $X$ relative to $X'$ over the period} $[0, \hor]$ via
\[
\rrh (X|X') \dfn \limsup_{t \to \infty} \pare{\frac{X_{\hor \wedge t} - X'_{\hor \wedge t}}{X'_{\hor \wedge t}}} = \limsup_{t \to \infty} \pare{\frac{X_{\hor \wedge t}}{X'_{\hor \wedge t}}} - 1.
\]
(The convention $0/0 = 1$ is used throughout.) In other words, $\rrh (X|X') = (X_\hor - X'_\hor) / X'_\hor$ holds on the event $\set{\hor < \infty}$, while $\rrh (X|X') = \limsup_{t \to \infty} \pare{ (X_t - X'_t) / X'_t} = \rri(X|X')$ holds on the event $\set{\hor = \infty}$. The above definition conveniently covers both cases. Observe that $\rrh(X|X')$ is a $[-1, \infty]$-valued random variable. Therefore, for any stopping time $\hor$ and $X,X' \in \X$, the quantity
\[
\errh(X|X') \dfn \expec \bra{\rrh(X|X')}
\]
is well defined and $[-1, \infty]$-valued. $\errh(X|X')$ represents the \textsl{expected return of $X$ relative to $X'$ over the time-period $[0, \hor]$}.

The concept of expected relative returns is introduced for purposes of portfolio selection. A first idea that comes to mind is to proclaim that $X' \in \X$ is ``strictly better'' than $X \in \X$ for investment over the period $[0, \hor]$ if $\errh(X'|X) > 0$. However, this is not an appropriate notion: it is easy to construct examples where both $\errh(X'|X) > 0$ and $\errh(X|X') > 0$ hold. The reason is that, in general, $\rrh(X|X') \neq - \rrh(X'|X)$. In fact, Proposition \ref{prop: err} below implies that $\rrh(X|X') \geq - \rrh(X'|X)$, with equality holding only on the event $\set{\lim_{t \to \infty} ( X_{\hor \wedge t} / X'_{\hor \wedge t}) = 1}$. A more appropriate definition would call $X' \in \X$ ``strictly better'' than $X \in \X$ for investment over the period $[0, \hor]$ if \emph{both} $\errh(X'|X) > 0$ and $\errh(X|X') \leq 0$ hold. In fact, because of the inequality $\rrh(X|X') \geq - \rrh(X'|X)$, $\errh(X|X') \leq 0$ is enough to imply $\errh(X'|X) \geq 0$, and one has $\errh(X'|X) > 0$ in the case where $\prob \bra{ \lim_{t \to \infty} \pare{X_{\hor \wedge t} / X'_{\hor \wedge t}} = 1} < 1$.

The discussion of the previous paragraph can be summarised as follows: while positive expected returns of $X \in \X$ with respect to $X' \in \X$ do not imply that $X$ is a better investment than $X'$, we may regard non-positive expected returns of $X \in \X$ with respect to $X' \in \X$ to indicate that $X'$ is a better investment than $X$. Given the use of ``$\limsup$'' in the equality $\rrh (X|X') = \limsup_{t \to \infty} \pare{ (X_t - X'_t) / X'_t}$, valid on $\set{\hor = \infty}$, it seems particularly justified to regard $X'$ as better than $X$ when $\erri(X|X') \leq 0$ holds, at least in an asymptotic sense. We are led to the following concept.

\begin{defn}\label{defn: num_prop}
We say that $X'$ has the \textsl{\num \ property} in a certain class of wealth processes for investment over the period $[0, \hor]$ if and only if $\errh(X|X') \leq 0$ holds for all other $X$ in the same class.
\end{defn}

\begin{rem} \label{rem: depend_on_hor}
The above definition is close in spirit to the \num\ in \cite{Long:1990hm}. However following closely \cite{Long:1990hm} and the results pertaining to the non-constrained case, one may be tempted to define the \num\ portfolio in a certain class of wealth processes by postulating that all other wealth processes in this class are supermartingales in units of the \num. However, in the context of drawdown constraints this would be a void concept as portfolios with the \num \ property may depend on the planning horizon---see Example \ref{exa: time-horizon_matters}.
\end{rem}

The next result contains some useful properties of (expected) relative returns. In particular, it implies that the terminal value of an investment with the \num \ property within a certain class of processes for investment over a specified period of time is essentially unique.

\begin{prop} \label{prop: err}
For any stopping time $\hor$, any $X \in \X$ and any $X' \in \X$, it holds that
\[
\rrh(X'|X) \geq - \frac{\rrh(X|X')}{1 + \rrh(X|X')} \geq - \rrh(X|X'),
\]
with equality on $\set{\hor < \infty}$. Furthermore, the following equivalence is valid:
\[
\errh(X'|X) \leq 0 \text{ and } \errh(X|X') \leq 0 \ \Longleftrightarrow \ \prob \bra{ \lim_{t \to \infty} \pare{\frac{X_{\hor \wedge t}}{X'_{\hor \wedge t}}} = 1} = 1.
\]
\end{prop}

The proof is reported in Subsection \ref{subsec: proof of prop: err} in the Appendix \ref{sec: proofs}. It follows from the above Proposition that if $\errh(X|X') \leq 0$ and $\errh(X'|X) \leq 0$ both hold, then $X_\hor = X'_\hor$ a.s.\ on $\set{\hor < \infty}$, while $\lim_{t \to \infty} ( X_{t} / X'_{t}) = 1$ a.s.\ on $\set{\hor = \infty}$, the latter being a version of ``asymptotic equivalence'' between $X$ and $X'$. 

The next result establishes existence of a process with the \num \ property in the class $\Xa$ sampled at $\hor$ for all $\alpha \in \zo$ and finite time-horizon $\hor$, and shows that such process is uniquely defined on the stochastic interval $\dbra{0, \hor} = \set{(\omega, t) \in \Omega \times \Real_+ \such 0 \leq t \leq T(\omega) }$. (Note that the latter uniqueness property is stronger than plain uniqueness of the terminal value of processes with the \num \ property that is guaranteed by Proposition \ref{prop: err}.) Theorem \ref{thm: num} later will address the possibility of an infinite time-horizon. 

\begin{thm} \label{thm: exist}
Let $\hor$ be a stopping time with $\prob[T < \infty] = 1$. Under condition (A1) of Assumption \ref{ass: basic}, there exists $\xt \in \Xa$, which may depend on $\hor$, such that $\errh(X | \xt) \leq 0$ holds for all $X \in \Xa$. Furthermore, $\xt$ has the following uniqueness property: for any other process $\zt \in \Xa$ such that $\errh(X | \zt) \leq 0$ holds for all $X \in \Xa$, $\xt = \zt$ a.s. holds on $\dbra{0, \hor}$. 
\end{thm}

The proof of Theorem \ref{thm: exist} is given at Subsection \ref{subsec: proof of thm:exist} of Appendix \ref{sec: proofs}.

\begin{rem} \label{rem: num implies log}
In the notation of Theorem \ref{thm: exist}, the log-utility maximisation problem at time $\hor$ is solved by the wealth process $\xta$. Indeed, the inequality $\log(x) \leq x - 1$, valid for all $x \in \Real_+$, gives
\[
\expec \bra{ \log \pare{ \frac{X_\hor}{\xt_\hor} } } \leq \expec \bra{ \frac{X_\hor}{\xt_\hor} - 1 } = \errh(X|\xt) \leq 0,
\]
for all $X \in \Xa$. This is a version of \emph{relative} expected log-optimality, which turns to actual expected log-optimality as soon as the expected log-maximisation problem is well-posed---in this respect, see also \cite[Subsection 3.7]{MR2335830}.

In view of Theorem \ref{thm: exist}, the above discussion ensures existence and uniqueness of expected log-utility optimal wealth processes for finite time-horizons in a drawdown-constrained investment framework. To the best of the authors' knowledge, results regarding existence and uniqueness of optimal processes for utility maximisation problems involving finite time-horizon and drawdown constraints are absent from the literature.
\end{rem}

\subsection{The \num \ property at times of maximum of the \num \ portfolio}

When $\alpha = 0$, the fact that $X / \Xhat$ is a nonnegative supermartingale and the optional sampling theorem imply that $\errh(X|\Xhat) \leq 0$ holds for all stopping times $\hor$ and all $X \in \X$. Therefore, the process $\Xhat$ has a ``global'' (in time) \num \ property. Furthermore, the supermartingale convergence theorem implies that $\lim_{t \to \infty} (X_t / \Xhat_t)$ $\prob$-a.s. exists for all $X \in \X$; therefore,
\begin{equation} \label{eq: pos_supermart_conv}
\rri (X|\Xhat) = \lim_{t \to \infty} \pare{ \frac{X_t - \xhat_t}{\xhat_t}} = \lim_{t \to \infty} \pare{ \frac{X_t}{\xhat_t}} - 1.
\end{equation}
For finite time-horizons, the situation is more complicated for $\alpha \in (0,1)$. In Theorem \ref{thm: num}, we shall see that $\xha$ has the \num \ property in $\Xa$ for certain stopping times (which include the asymptotic case $\hor = \infty$). However, $\xha$ does not have the \num \ property for all finite time-horizons, as the Example \ref{exa: time-horizon_matters} shows. This fact motivates the statement of Theorem \ref{thm: exist}, where it is hinted that portfolios with the \num \ property may depend on the time-horizon---see Remark \ref{rem: depend_on_hor}.

\begin{exa} \label{exa: time-horizon_matters}
Fix $\alpha \in (0,1)$. Define $T \dfn \inf \{ t \in (0,\infty) \such \Xhat_t / \Xhat^*_t = \alpha \}$ and observe that Proposition \ref{prop: dd_osc} implies that $\prob[T < \infty] = 1$ holds. With $\Xhat^{T}$ denoting the process $\Xhat$ stopped at $T$, we have $\Xhat^{T} \in \Xa$. The \num \ property of $\Xhat$ in $\X \supseteq \Xa$ implies that $\err_{T} (X  | \Xhat^{T}) \leq 0$ for all $X \in \Xa$, resulting in the \num \ property of $\Xhat^{T}$ in $\Xa$ over the investment period $[0, T]$. Since $\prob [ \xha_{T} = \Xhat_{T} ] = 0$, it follows that $\xha_{T}$ fails to have the \num \ property in $\Xa$ over the investment period $[0, T]$.

Before abandoning this example, note that if one follows the non-constrained \num \ portfolio $\Xhat$ up to $T$, the drawdown constraints will mean that one has to invest all capital in the baseline account from time $T$ onwards. It is clear that this strategy will not be long-run optimal.
\end{exa}

We continue with a definition of a class of stopping times which will be important in the sequel.

\begin{defn}
A stopping time $\tau$ will be called a \textsl{time of maximum of $\Xhat$} if $\Xhat_\tau = \Xhat^*_\tau$ holds a.s. on the event $\set{\tau < \infty}$.
\end{defn}

A couple of remarks are in order. Firstly, from \eqref{eq: xa_defn} one can immediately see that times of maximum of $\Xhat$ are also times of maximum of $\xha$ for all $\alpha \in \zo$. Secondly, the restriction in the definition of a time $\tau$ of maximum of $\Xhat$ is only enforced on $\set{\tau < \infty}$. Under Assumption \ref{ass: basic}, and in view of Theorem \ref{thm: asss}, one has $\Xhat_\tau = \Xhat^*_\tau = \infty$ holding a.s. on $\set{\tau = \infty}$. For this reason, $\tau = \infty$ is an important special case of a time of maximum of $\Xhat$.

The following theorem, the second main result of this section, establishes the \emph{{\num}  property} of $\xha$ in $\Xa$ over $[0,\infty]$ or, more generally, over $[0,\tau]$ for any time $\tau$ of maximum of $\Xhat$. We recall that $\Xa = \set{\xa \such X \in \X}$.

\begin{thm} \label{thm: num}
Recall that $\xha \in \Xa$ is defined from $\Xhat$ via \eqref{eq: xa_defn}. Under Assumption \ref{ass: basic}, for any $\alpha\in [0,1)$ and $X \in \X$, we have:
\begin{enumerate}
	\item $\lim_{t \to \infty} (\xa_t / \xha_t)$ a.s. exists in $\Real_+$. Moreover,
\begin{equation} \label{eq: asympt num rel}
\rri (\xa|\xha) = 
\pare{ \lim_{t \to \infty} \pare{ \frac{X_t}{\Xhat_t}} }^{1-\alpha} - 1 = \pare{1 + \rri (X|\Xhat)}^{1- \alpha} - 1.
\end{equation}
	\item For $\sigma$ and $\tau$ two times of maximum of $\Xhat$ with $\sigma \leq \tau$, it a.s. holds that
\begin{equation} \label{eq: supermart_analogue}
\expec \bra{\rrt(\xa|\xha) \ \big| \ \F_\sigma}\leq \rr_{\sigma}(\xa|\xha).
\end{equation}
In particular, letting $\sigma=0$, $\errt(Z|\xha) \leq 0$ holds for any $\alpha \in \zo$ and $Z \in \Xa$.
\end{enumerate}
\end{thm}

We proceed with several remarks on the implications of Theorem \ref{thm: num}, the proof of which is given in Subsection \ref{subsec: proof of thm:num} of Appendix \ref{sec: proofs}.

\begin{rem}
The existence of the limit in \eqref{eq: pos_supermart_conv} is guaranteed by the nonnegative supermartingale convergence theorem. In contrast, proving that $\lim_{t \to \infty} (\xa_t / \xha_t)$ exists a.s.\ for $X \in \X$ and $\alpha \in (0,1)$ is more involved, since, in general, the process $\xa / \xha$ does not have the supermartingale property. In fact, the existence of the latter limit is proved together with the asymptotic relationship \eqref{eq: asympt num rel}. Note, however, that an analogue of the supermartingale property is provided by statement (2) of Theorem \ref{thm: num}. Indeed, \eqref{eq: supermart_analogue} implies that, when sampled at an increasing sequence of times of maximum of $\Xhat$, the process $\xa / \xha$ is a supermartingale along these times for all $X \in \X$.
\end{rem}

\begin{rem}
Given statement (1) of Theorem \ref{thm: num}, the fact that $\erri(\xa|\xha) \leq 0$ holds for any  $X \in \X$  and $\alpha \in \zo$ is a simple consequence of Jensen's inequality. Indeed, for any $X \in \X$,
\begin{align*}
    \erri (\xa|\xha) &= \expec \bra{ \pare{1 + \rri (X|\Xhat)}^{1- \alpha}} - 1 \\
&\leq  \pare{\expec \bra{1 + \rri (X|\Xhat)}}^{1- \alpha} - 1 \\
&=  \pare{ 1 + \erri (X|\Xhat)}^{1- \alpha} - 1 \ \leq \ 0.
\end{align*}
The full proof of statement (2) of Theorem \ref{thm: num}, given in Appendix \ref{sec: proofs}, is more involved.
\end{rem}

\begin{rem}
The fact that $\rri (\xa|\xha) = \big( 1 + \rri (X|\Xhat) \big)^{1- \alpha} - 1$ holds for all $\alpha \in \zo$ can be easily seen to imply that $\big| \rri (\xb|\xhb) \big| \leq \big| \rri (\xa|\xha)\big|$ holds whenever $0 \leq \alpha \leq \beta < 1$. In other words, using the same generating wealth process $X$ and enforcing harsher drawdown constraints reduces the (asymptotic) difference in the performance of the drawdown-constrained process $\xa$ against the long-run optimum $\xha$.
\end{rem}

\begin{rem} 
Let us consider the hitting times of $\Xhat$, parametrized on the logarithmic scale:
\begin{equation} \label{eq: level_crossing}
\tau_\ell \dfn \inf \set{t \in \Real_+ \such \Xhat_t = \el},\quad \ell \in \Real_+.
\end{equation}
Note that $\tau_\ell$ is a time of maximum of $\Xhat$. Since times of maximum of $\Xhat$ coincide with times of maximum of $\xha$ for $\alpha \in \zo$, $\tau_\ell = \inf \big\{ t \in \Real_+ \such \xha_t = \ela \big\}$ holds for all $\alpha \in \zo$.
According to Assumption \ref{ass: basic}, $\prob \bra{\tau_\ell < \infty} = 1$ holds for all $\ell \in \Real_+$.

By Remark \ref{rem: num implies log}, the log-utility maximisation problem at time $\tau_\ell$ for the class $\Xa$ is solved by the wealth process $\xha$. Moreover, assume that $U: \Real_+ \mapsto \Real \cup \set{- \infty}$ is any increasing and concave function such that $U(x) > - \infty$ for all $x \in (0, \infty)$. Jensen's inequality implies that
\[
\expec \bra{U(\xa_{\tau_\ell})} \leq U(\expec \bra{\xa_{\tau_\ell}}) \leq U(\ela) = \expec \bra{U(\xha_{\tau_\ell})}, \quad \text{for all } X \in \X.
\]
It follows that \emph{any} (and not only the logarithmic) utility maximisation problem at a hitting time $\tau_\ell$ for the class $\Xa$ is solved by the wealth process $\xha$. This is a remarkable fact that is extremely robust, since it does not require any modelling assumptions.
\end{rem}

\section{More on Asymptotic Optimality}
\label{sec: more_asympt_optimal}

\subsection{Maximisation of long-term growth} \label{subsec: asympt_growth}

The next theorem is concerned with the asymptotic growth-optimality property of $\xha$ in $\Xa$ for $\alpha \in \zo$. It extends the result of \cite[Section 7]{CvitanicKaratzas:94} to a  more general setting and with a simpler proof. In the subsequent subsection we continue with a considerably finer analysis relating the finite-time and asymptotic optimality of $\xha$ in $\Xa$. 

One of the equivalent conditions to (A1) of Assumption \ref{ass: basic} is that a market-growth process $G$ exists: $G$ is a non-negative and non-decreasing process such that $\log(\Xhat)=G+L$ for a local martingale $L$; furthermore,  Assumption (A2) is equivalent to $\lim_{t \to \infty} G_t = \infty$; see Theorem \ref{thm: asss}. We can use the process $G$ to control the growth rate of any portfolio.

\begin{thm} \label{thm: asympt_growth}
Under Assumption \ref{ass: basic}, for any $Z \in \Xa$ we a.s.\ have that
\begin{equation}\label{eq: growth_ineq}
\limsup_{t \to \infty} \pare{\frac{1}{G_t} \log(Z_t)} \leq 1 - \alpha = \lim_{t \to \infty} \pare{\frac{1}{G_t} \log\pare{\xha_t}}.
\end{equation}
\end{thm}
The proof is reported in Subsection \ref{subsec: proof of thm: asympt_growth} in Appendix \ref{sec: proofs}.
\begin{rem}
Fix $\alpha \in \zo$. In the setting of Theorem \ref{thm: asympt_growth}, any $\xa \in \Xa$ such that, a.s.,
\[
\lim_{t \to \infty} \pare{\frac{1}{G_t} \log \pare{ \frac{\xa_t}{\xha_t} } } = 0
\]
also enjoys the asymptotic growth-optimality property in the sense of achieving equality in \eqref{eq: growth_ineq}. As a simple example, let $X\in \X$, $\kappa\in (0,1)$ and $\xt :=\kappa \Xhat + (1-\kappa)X$. Then $\xt^*\geq \kappa \Xhat^*$ so that $\xta\geq \alpha \xta^*\geq \alpha (\kappa \Xhat^*)^{1-\alpha}= \pare{\alpha \kappa^{1-\alpha}} \xha^* \geq \pare{\alpha \kappa^{1-\alpha}} \xha$ and, consequently, $\xta$ enjoys the asymptotic growth optimality. In contrast, the asymptotic \num\ property is much stronger. Combining Theorem \ref{thm: num} and Proposition \ref{prop: err}, it follows that if $\xa \in \Xa$ is to have the asymptotic \num \ property, then the much stronger ``asymptotic equivalence'' condition $\lim_{t \to \infty} \big( \xa_t / \xha_t \big) = 1$ has to be a.s. valid. We shall see below that this in fact implies the even stronger condition $\xa=\xha$.
\end{rem}

\subsection{Optimality through sequences of stopping times converging to infinity} \label{subsec: turnpike}

By Theorem \ref{thm: num}, $\erri(\xa|\xha) \leq 0$ holds for all $X \in \X$, a result which can be interpreted as long-run \num \ optimality property of $\xha$ in $\Xa$. However, in effect, this result assumes that the investment time-horizon is actually equal to infinity. On both theoretical and practical levels, one may be rather interested in considering a sequence of stopping times $(T_n)_{\nin}$ that converge to infinity and examine the behaviour of optimal wealth processes (in the \num \ sense) in the limit. We present two results in this direction. Proposition \ref{prop: re_asympt_num} establishes that the only process in $\Xa$ possessing the \num \ property along an increasing sequence of stopping times tending to infinity is $\xha$. The second result, Theorem \ref{thm: turnpike}, is more delicate than Proposition \ref{prop: re_asympt_num}, and may be regarded as a version of so-called \emph{turnpike theorems}, an appellation coined in \cite{leland1972turnpike}. While the traditional formulation of turnpike theorems involves two investors with long financial planning horizon and similar preferences for large levels of wealth, Theorem \ref{thm: turnpike} compares a portfolio having the \num \ property for a long, but finite, time-horizon  with the corresponding portfolio having the asymptotic \num \ property. Loosely speaking, Theorem \ref{thm: turnpike} states that, when the time horizon $T$ is long, the process $\xta$ that has the \num \ property in $\Xa$ for investment over the interval $[0, T]$ will be very close initially (in time) to $\xha$ in a very strong sense.

\begin{prop} \label{prop: re_asympt_num}
Under the validity of Assumption \ref{ass: basic}, suppose that there exist $X \in \X$ and a sequence of (possibly infinite-valued) stopping times $(T_n)_{\nin}$ with $\limn \prob \bra{T_n > t} = 1$ holding for all $t \in \Real_+$, such that $\liminf_{n \to \infty} \err_{T_n} (\xha|\xa) \leq 0$. Then, $\xa = \xha$.
\end{prop}

The proof is given in Subsection \ref{subsec: proof of prop: re_asympt_num} in Appendix \ref{sec: proofs}. 
Following the reasoning therein, one can also show that if $\tau$ is a time of maximum of $\Xhat$ and $\errt(\xha|\xa) \leq 0$ holds for some $X \in \X$, then $\xa = \xha$ holds identically on the stochastic interval $\dbra{0, \tau}$.

In order to state Theorem \ref{thm: turnpike}, we define a strong notion of convergence in the space of semimartingales, introduced in \cite{MR544800}.

\begin{defn} \label{dfn: emery}
For a stopping time $\hor$, we say that a sequence $(\xi^n)_{\nin}$ of semimartingales converges over $[0, \hor]$ in the \textsl{Emery topology} to another semimartingale $\xi$, and write $\Shn \xi^n = \xi$, if
\begin{equation} \label{eq: emery_conv}
\lim_{n \to \infty} \sup_{\eta \in \Pre_1} \prob \bra{  \sup_{t \in [0, \hor]} \abs{\eta_0 (\xi^n_0 - \xi_0) + \int_0^t \eta_s \ud \xi^n_s - \int_0^t \eta_s \ud \xi_s} > \epsilon } = 0
\end{equation}
holds for all $\epsilon > 0$, where $\Pre_1$ denotes the set of all predictable processes $\eta$ with $\sup_{t \in \Real_+} |\eta_t| \leq 1$. Furthermore, we say that the sequence $(\xi^n)_{\nin}$ of semimartingales converges \textsl{locally in the Emery topology} to another semimartingale $\xi$, and write $\Sln \xi^n = \xi$, if $\Shn \xi^n = \xi$ holds for all a.s. finitely-valued stopping times $\hor$.
\end{defn}

\begin{rem} \label{rem: ucp_conv}
In the setting of Definition \ref{dfn: emery}, assume that $(\xi^n)_{\nin}$ converges locally in the Emery topology to $\xi$. By taking $\eta \equiv 1$ in \eqref{eq: emery_conv}, we see that
\[
\lim_{n \to \infty} \prob \bra{  \sup_{t \in [0, \hor]} \abs{\xi^n_t - \xi_t} > \epsilon } = 0
\]
holds for all $\epsilon > 0$ and all a.s. finitely-valued stopping times $\hor$. In other words, the sequence $(\xi^n)_{\nin}$ converges \textsl{in probability, uniformly on compacts,} to $\xi$.
\end{rem}

\begin{thm} \label{thm: turnpike}
Suppose that $(T_n)_{\nin}$ is a sequence of stopping times such that $\limn \prob \bra{T_n > t} = 1$ holds for all $t \in \Real_+$. For each $\nin$, let $\xta^n \in \Xa$ have the \num \ property in $\Xa$ for investment over the period $[0, T_n]$. Under Assumption \ref{ass: basic}, it holds that $\Sln \xta^n = \xha$.
\end{thm}

The proof of Theorem \ref{thm: turnpike} is given in Subsection \ref{subsec: proof of thm:turnpike} of Appendix \ref{sec: proofs}.

\begin{rem}
In the setting of Theorem \ref{thm: turnpike}, the fact that $\Sln \xta^n = \xha$ implies by Proposition 2.9 in \cite{Kar11} that $\lim_{n \to \infty} \prob \big[ \, [ \xta^n - \xha, \xta^n - \xha]_\hor > \epsilon \big] = 0$ holds for all a.s. finitely-valued stopping times $\hor$ and $\epsilon > 0$. Writing $\xha = 1 + \int_0^\cdot (\hha_t, \ud S_t)$ and $\xta^n = 1 + \int_0^\cdot (\hta^n_t, \ud S_t)$ for all $\nin$ for appropriate $d$-dimensional strategies $\hha$ and $(\hta^n)_{\nin}$, we obtain
\[
\lim_{n \to \infty} \prob \bra{ \int_0^\hor \pare{\hta^n_t - \hha_t, \, \ud [S, S]_t \big(\hta^n_t - \hha_t\big)} > \epsilon } = 0,
\]
for all a.s. finitely-valued stopping times $\hor$ and $\epsilon > 0$. The previous relation implies that it is not only wealth that converges to the limiting one in each finite time-interval---the corresponding employed strategy does so as well.
\end{rem}

\begin{rem} \label{rem: careful_turnpike}
In the setting of Theorem \ref{thm: turnpike}, the conclusion is that convergence of $\xta^n$ to $\xha$ holds over finite time-intervals that do \emph{not} depend on $\nin$. One can ask whether the whole wealth process $\xta^n$ is close to $\xha$ over the stochastic interval $\dbra{0, T_n}$ for each $\nin$. This is not true in general; in Appendix \ref{sec: careful_turnpike} we present an example, valid under \emph{all} models for which Assumption \ref{ass: basic} holds, where the ratio $\xta^n_{T_n} / \xha_{T_n}$ as $n \to \infty$ oscillates between $1/(2 - \alpha)$ and $\infty$. Note that the example only covers cases where $\alpha \in (0, 1)$; if $\alpha = 0$, $\xta^n = \xha$ always holds for all $\nin$.
\end{rem}

\section{Conclusions}
The \num\ portfolio $\Xhat$, in the global sense of condition (B1) in Theorem \ref{thm: asss_num}, exists and is unique in a very general modelling set-up. The \num\ property is a strong one and implies that $\Xhat$ maximises the growth rate as well as the expected logarithmic utility, see Remarks \ref{rem: num_log}, \ref{rem: num implies log} and Theorem \ref{thm: asympt_growth}. Many experts on portfolio allocation have argued  that it makes it a natural choice for a long-run investor. On the other hand, one has to admit that it offers little flexibility to control for investor's risk appetites. The literature usually points to expected utility (including ``quadratic utility'' reflecting in some sense Markowitz's mean-variance approach) as a more flexible framework. 

However, expected utility maximisation has its inherent problems: its solution typically depends on the (arbitrary) choice of a time horizon, the utility function is a theoretical concept hard to elucidate and the optimal solution depends in a complex way on the specification of the model and the preferences. In this paper we suggest a possible way out of this unsatisfactory situation: we propose maximising the growth rate within a restricted class of investment strategies. We only consider wealth processes $X$ that satisfy a drawdown constraint: $X\geq \alpha X^*$, with $\alpha\in [0,1)$ which quantifies the investor's attitude against risk. Drawdown constraints are encountered in practice and were shown to be an effective and robust way of encoding preferences, at least for long horizons. Their drawback lies in the path dependent, non-myopic, nature of the  constraint, which renders certain features of traditional asset-allocation theory invalid.

This paper presented a rather complete investigation of the \num\ property in a drawdown-constrained context. First, we gave a new definition based on the expected relative return, which extends the \num\ property from a global setting $\X$ to any subset of investment strategies. We showed that for each time horizon there exists an essentially unique portfolio with the \num\ property within the class $\Xa$. Moreover, as horizon became distant, these are close, in a very strong sense on any fixed time interval to $\xha$, which is the unique portfolio with the asymptotic \num\ property. It is defined through an explicit and model-independent Az\'ema--Yor transformation \eqref{eq: xa_defn} from the global \num\ portfolio $\Xhat$ and has a natural investment interpretation, see Remark \ref{rem: SDEintepret}. Furthermore, $\xha$ has the \num\ property also along an increasing sequence of stopping times: the times of maximum of $\Xhat$. However, contrary to the unconstrained case, it does not  enjoy the \num\ property for all times. This is an important novel feature. 

\appendix

\section{Proofs} \label{sec: proofs}

\renewcommand{\theequation}{App.\arabic{equation}}    
 \setcounter{equation}{0}  
  
We start by describing in Subsection \ref{subsec: equivalentass} several useful equivalent formulations of Assumption \ref{ass: basic}.
Thereafter, through the course of Appendix \ref{sec: proofs}, the validity of Assumption \ref{ass: basic} is always in force. The only exception are Subsection \ref{subsec: proof of prop: err}, where no assumption is made, and Subsection \ref{subsec: proof of thm:exist}, where only the condition (A1) of Assumption \ref{ass: basic} is required.

\subsection{Equivalent conditions to Assumption \ref{ass: basic}}\label{subsec: equivalentass}
Recall the market specification in Section \ref{subsec: model}. For $i \in \set{1, \ldots, d}$ write $S^i = S^i_0 + B^i + M^i$ for the Doob-Meyer decomposition of $S^i$ into a continuous finite variation process $B^i$ with $B^i_0 = 0$ and a local martingale $M^i$ with $M^i_0 = 0$. For $i \in \set{1, \ldots, d}$ and $j \in \set{1, \ldots, d}$, $[S^i, S^j] = [M^i, M^j]$ denotes the covariation process of $S^i$ and $S^j$.

The following result follows Theorem 4 of \cite{Kar_PF} and contains useful equivalent conditions to the ones presented in Assumption \ref{ass: basic}.

\begin{thm} \label{thm: asss}

Condition (A1) of Assumption \ref{ass: basic} is equivalent to any of the following:
\begin{enumerate}
	\item[(B1)] There exists $\Xhat  \in \X$ such that $X/ \Xhat$ is a (nonnegative) local martingale for all $X \in \X$.
	\item[(C1)] There exists a $d$-dimensional process $\rho$ such that $B^i = \int_0^\cdot \sum_{j=1}^d \rho^j_t \ud [S^j, S^i]_t$ holds for each $i \in \set{1, \ldots, d}$. Furthermore, the nonnegative and nondecreasing process
\begin{equation} \label{eq: growth}
G \dfn \frac{1}{2} \int_0^\cdot  \pare{\rho_t,  \ud [S, S]_t \rho_t} \equiv \frac{1}{2} \int_0^\cdot \sum_{i=1}^d \sum_{j=1}^d  \rho^i_t \rho^j_t \ud [S^j, S^i]_t 
\end{equation}
is such that $\prob \bra{G_T < \infty} = 1$ holds for all $T \in \Real_+$.
\end{enumerate}
Under the validity of any of (A1), (B1), (C1), and with the above notation, it holds that
\begin{equation} \label{eq: num}
\log(\Xhat) = G + L, \quad \text{where } L \dfn \int_0^\cdot \sum_{i=1}^d \rho^i_t \ud M^i_t
\end{equation}

Furthermore, under the validity of any of the equivalent (A1), (B1), (C1), condition (A2) of Assumption \ref{ass: basic} is equivalent to any of the following:
\begin{enumerate}
	\item[(B2)] $\prob \big[ \lim_{t \to \infty} \Xhat_t = \infty \big] = 1$.
	\item[(C2)] $\prob \big[ G_\infty = \infty \big] = 1$, where $G_\infty \dfn \, \uparrow \lim_{t \to \infty} G_t$.
\end{enumerate}
\end{thm}

\begin{proof}
The fact that the three conditions (A1), (B1) and (C1) are equivalent, as well as the validity of \eqref{eq: num}, can be found in \cite[Theorem 4]{Kar_PF}. Now, assume any of the equivalent conditions (A1), (B1) or (C1). Clearly, (B2) implies (A2). On the other hand, suppose that there exists $X \in \X$ such that $\prob \bra{\lim_{t \to \infty} X_t = \infty} = 1$. The nonnegative supermartingale theorem implies that $\lim_{t \to \infty} \big(X_t / \Xhat_t)$ $\prob$-a.s. exists in $\Real_+$, which implies that $\prob \big[ \lim_{t \to \infty} \Xhat_t = \infty \big] = 1$ holds as well. Therefore, (A2) implies (B2). Continuing, note that
\eqref{eq: num} implies that
\[
[L, L] =  \int_0^\cdot  \pare{\rho_t, \ud [M, M]_t \rho_t} =  \int_0^\cdot  \pare{\rho_t, \ud [S, S]_t \rho_t} = 2 G.
\]
In view of the celebrated result of Dambis, Dubins and Schwarz---see Theorem 3.4.6 in \cite{MR1121940}---there exists a standard Brownian motion $\beta$ (in a potentially enlarged probability space, and the Brownian motion property of $\beta$ is with respect to its own natural filtration) such that $L_t = \beta_{2 G_t}$ holds for $t \in \Real_+$. It follows that $\log(\Xhat_t) = G_t + \beta_{2 G_t}$ holds for $t \in \Real_+$. Therefore, on $\set{G_\infty < \infty}$, $\lim_{t \to \infty} \Xhat_t$ a.s. exists and is $\Real_+$-valued. On the other hand, the strong law of large numbers for Brownian motion implies that on $\set{G_\infty = \infty}$, $\lim_{t \to \infty} \big( \log \big( \Xhat_t \big) / G_t\big) = 1$ a.s. holds, which in turn implies that $\lim_{t \to \infty} \Xhat_t = \infty$ a.s. holds. The previous facts imply the a.s. set-equality $\set{G_\infty = \infty} = \{ \lim_{t \to \infty} \Xhat_t = \infty \}$, which establishes the equivalence of conditions (B2) and (C2) and completes the proof.
\end{proof}

\begin{rem} 
In It\^o process models, it holds that $B^i = \int_0^\cdot S^i_t b^i_t \ud t$ and $M^i = \int_0^\cdot S^i_t \sum_{j=1}^m \sigma^{ij}_t \ud W^j_t$ for $i \in \set{1, \ldots, d}$, where $b = (b^1, \ldots, b^d)$ is the predictable $d$-dimensional vector of excess rates of return, $(W^1, \ldots, W^m)$ is an $m$-dimensional standard Brownian motion, and we write $c = \sigma \sigma^\top$ for the predictable $d \times d$ matrix-valued process of local covariances. According to Theorem \ref{thm: asss}, condition (A1) of Assumption \ref{ass: basic} is equivalent to the fact that there exists a $d$-dimensional process $\rho$ such that $c  \rho = b$, in which case we write $\rho = c^\dagger b$ where $c^\dagger$ is the Moore-Penrose pseudo-inverse of $c$, and that $G \dfn (1/2) \int_0^\cdot (b_t, c^\dagger_t b_t) \ud t = (1/2) \int_0^\cdot (\rho_t, c_t \rho_t) \ud t$ is an a.s. finitely-valued process. Observe that the process $G$ is half of the integrated squared risk-premium in the market.
\end{rem}

\subsection{Proof of Proposition \ref{prop: dd_charact}}

\label{subsec: proof of prop: dd_charact}

Since $\set{\xa \such X \in \X} \subseteq \Xa$ has already been established, we only need to show that $\Xa \subseteq \set{\xa \such X \in \X}$ also holds. Pick any $\chi \in \Xa$ and define
\[
X \dfn \frac{1}{1 - \alpha} (\chi^*)^{\alpha / (1 - \alpha)} \chi - \frac{\alpha}{1 - \alpha} (\chi^*)^{1 / (1 - \alpha)} = \frac{1}{1 - \alpha} (\chi^*)^{\alpha / (1 - \alpha)} \pare{\chi - \alpha \chi^*}.
\]
The fact that $\chi / \chi^* \geq \alpha$ implies that $X \geq 0$. Furthermore, and since $\int_{0}^\infty \indic_{\set{\chi_t < \chi^*_t}} \ud \chi^*_t = 0$ a.s. holds, a use of It\^o's formula gives
\[
X = 1 + \int_0^\cdot \frac{1}{1 - \alpha} (\chi^*_t)^{\alpha/(1 - \alpha)} \ud \chi_t,
\]
which implies that $X \in \X$. Finally using the fact that $\chi$ and $X$ have the same times of maximum---which implies, in particular, that $\chi^* = (X^*)^{1 - \alpha}$---it is straightforward to check that $\chi = \xa$. Therefore, $\Xa \subseteq \set{\xa \such X \in \X}$ and the proof of Proposition \ref{prop: dd_charact} is complete.

\subsection{Proof of Proposition \ref{prop: dd_osc}}

\label{subsec: proof of prop: dd_osc}

Since $\xha / \xha^* = \alpha + (1 - \alpha) (\Xhat / \Xhat^*)$ holds in view of \eqref{eq: dd_xa}, we only need to establish that $0 = \liminf_{t \to \infty} (\Xhat_t / \Xhat^*_t) < \limsup_{t \to \infty} (\Xhat_t / \Xhat^*_t) = 1$. The fact that $\limsup_{t \to \infty} (\Xhat_t / \Xhat^*_t) = 1$ follows directly from $\lim_{t \to \infty} \Xhat_t = \infty$. On the other hand, the fact that $\liminf_{t \to \infty} (\Xhat_t / \Xhat^*_t) = 0$ follows immediately from the next result (which is stated separately as it is also used on another occasion) and the martingale version of the Borel-Cantelli lemma.

\begin{lem} \label{lem: fin_dd_time}
Let $\sigma$ be a stopping time with $\prob[\sigma < \infty] = 1$. For $\alpha \in (0,1)$ define the stopping time $T \dfn \inf \big \{ t \in (\sigma, \infty) \such \Xhat_t / \Xhat^*_t \leq \alpha \big \}$. Then $\prob \bra{T < \infty} = 1$.
\end{lem}

\begin{proof}
Recall that $\lim_{t \to \infty} \Xhat_t = \infty$ holds by Theorem \ref{thm: asss}. Using the result of Dambis, Dubins and Schwarz---Theorem 3.4.6 in \cite{MR1121940}---and a time-change argument, \eqref{eq: num} implies that we can assume without loss of generality that $\Xhat$ satisfies $\Xhat_t = \exp\pare{t/2 + \beta_t}$ for $t \in \Real_+$, where $\beta$ is a standard Brownian motion. Furthermore, using again the fact that $\lim_{t \to \infty} \Xhat_t = \infty$, we may assume without loss of generality that $\sigma$ is a time of maximum of $\Xhat$. Then, the independent increments property of Brownian motion implies that we can additionally assume without loss of generality that $\sigma = 0$. Set $\sigma_0 =0$ and, via induction, for each $\nin$ set
\[
\sigma_n \dfn \inf \set{t \in (\sigma_{n-1}, \infty) \such \Xhat_t = \mathrm{e} \Xhat_{\sigma_{n-1}}}, \text{ and } T_n = \inf \set{ t \in (\sigma_{n-1}, \infty) \such \Xhat_t / \Xhat^*_t = \alpha }.
\]
With $T = T_1$, we wish to show that $\prob \bra{T < \infty} = 1$. For each $\nin$ define the event $A_n \dfn \set{T_n < \sigma_n}$.  Note that $\prob \bra{A_n \such \F_{\sigma_{n-1}}} = \prob \bra{A_1}$ holds for all $\nin$ in view of the regenerating property of Brownian motion and the fact that each $\sigma_{n-1}$, $\nin$, is a time of maximum of $\Xhat$. Since $\limsup_{n \to \infty} A_n \subseteq \set{T < \infty}$, the martingale version of the Borel-Cantelli lemma implies that $\prob[T < \infty] = 1$ will be established as long as we can show that $\prob[T_1 < \sigma_1] = \prob[A_1] > 0$.

Since $\int_{0}^\infty \indic_{\set{\Xhat_t < \Xhat^*_t}} \ud \Xhat^*_t = 0$ a.s. holds, It\^o's formula implies that
\[
\frac{\Xhat^*}{\Xhat} = 1 + \int_0^\cdot \Xhat^*_t \ud \pare{\frac{1}{\Xhat_t}} + \log(\Xhat^*).
\]
Both processes $\Xhat^* / \Xhat$ and $\log(\Xhat^*)$ are bounded on the stochastic interval $\dbra{0, \sigma_1 \wedge T_1}$---therefore, since $\prob \bra{\sigma_1 < \infty} = 1$ and $\int_0^\cdot \Xhat^*_t \ud \big( 1 / \Xhat_t \big)$ is a local martingale (by Assumption \ref{ass: basic} and the fact that $1 \in \X$), a localisation argument gives
\[
 \prob \bra{\sigma_1 \leq T_1} + \frac{1}{\alpha} \prob \bra{T_1 < \sigma_1} = \expec \bra{\frac{\Xhat^*_{\sigma_1 \wedge T_1}}{\Xhat_{\sigma_1 \wedge T_1}}} = 1 + \expec \bra{\log(\Xhat^*_{\sigma_1 \wedge T_1})} \geq 1 + \prob \bra{\sigma_1 \leq T_1},
\]
which gives $\prob \bra{T_1 < \sigma_1} \geq \alpha > 0$ and completes the proof of Lemma \ref{lem: fin_dd_time}.
\end{proof}

\subsection{Proof of Proposition \ref{prop: err}}\label{subsec: proof of prop: err}
To begin with, note that
\[
1 + \rrh(X | X') = \limsup_{t \to \infty} \pare{\frac{X_{\hor \wedge t}}{X'_{\hor \wedge t}}} \geq \pare{\limsup_{t \to \infty} \pare{\frac{X'_{\hor \wedge t}}{X_{\hor \wedge t}}}}^{-1} = \frac{1}{1 + \rrh(X' | X)},
\]
with equality holding on $\set{\hor < \infty}$. Continuing, we obtain
\[
\rrh(X | X') + \rrh(X' | X) \geq \frac{1}{1 + \rrh(X' | X)} - 1 + \rrh(X' | X) = \frac{\rrh(X' | X)^2}{1 + \rrh(X' | X)}.
\]
Upon interchanging the roles of $X$ and $X'$, we also obtain the corresponding inequality $\rrh(X | X') + \rrh(X' | X) \geq \rrh(X | X')^2/ \pare{1 + \rrh(X | X')}$; therefore,
\begin{equation} \label{eq: rrt_ineqs}
\rrh(X | X') + \rrh(X' | X) \geq  \frac{\rrh(X' | X)^2}{1 + \rrh(X' | X)} \vee \frac{\rrh(X | X')^2}{1 + \rrh(X | X')}.
\end{equation}
It immediately follows that $\rrh(X'|X) + \rrh(X|X') \geq 0$.  Therefore, by \eqref{eq: rrt_ineqs}, the conditions $\errh(X'|X) \leq 0$ and $\errh(X|X') \leq 0$ are equivalent to $\prob \bra{ \rrh(X | X') = 0 = \rrh(X' | X)} = 1$, which is in turn equivalent to $\prob \bra{ \lim_{t \to \infty} \pare{X_{\hor \wedge t} / X'_{\hor \wedge t}} = 1} = 1$.

\subsection{Proof of Theorem \ref{thm: exist}} \label{subsec: proof of thm:exist}
For the purposes of Subsection \ref{subsec: proof of thm:exist}, only condition (A1) of Assumption \ref{ass: basic} is in force.  Fix an a.s. finitely-valued stopping time $T$ throughout. As the result of Theorem \ref{thm: exist} for the case $\alpha = 0$ is known, we tacitly assume that $\alpha \in (0,1)$ throughout.

\subsubsection{Existence}
We shall first prove existence of a process with the \num \ property in $\Xa$ for investment over the period $[0, \hor]$. As $T$ is a.s. finitely-valued, without loss of generality we shall assume that all processes that appear below are constant after time $T$, and their value after time $T$ is equal to their value at time $T$. In particular, the limiting value of a process for time tending to infinity exists and is equal to its value at time $T$.

Define $\Xo$ as the class of all nonnegative \cadlag \ processes $Y$ with $Y_0 \leq 1$ and with the property that $Y X$ is a supermartingale for all $X \in \X$. Note that $(1 / \Xhat) \in \Xo$. In a similar way, define $\Xoo$ as the class of all nonnegative \cadlag \ processes $\chi$ with $\chi_0 \leq 1$ and with the property that that $Y \chi$ is a supermartingale for all $Y \in \Xo$. It is clear that $\X \subseteq \Xoo$. The next result reveals the exact structure of $\Xoo$.

\begin{thm}[Optional Decomposition Theorem \cite{MR1469917}, \cite{MR1651229}] \label{thm: ODT}
The class $\Xoo$ consists exactly of all processes $\chi$ of the form $\chi = X (1 - A)$, where $X \in \X$ and $A$ is an adapted, nonnegative and nondecreasing \cadlag \ process with $0 \leq A \leq 1$.
\end{thm}

The result that follows enables one to construct a process that will be a candidate to have the  \num \ property in $\Xa$ for investment over the interval $[0, T]$.

\begin{lem} \label{lem: xa conv and bdd}
For any $\alpha \in \zo$and $t \in [0, \infty]$, the set $\set{Z_t \such  Z \in \Xa}$ is convex and bounded in $\prob$-measure, the latter meaning that $\lim_{K\to\infty} \sup_{Z\in \Xa} \prob \bra{Z_t > K} = 0$.
\end{lem}

\begin{proof}
Fix $\alpha \in \zo$. Let $\lambda \in [0,1]$ and pick processes $X \in \X$ and $X' \in \X$. Since $\X$ is convex, $((1 - \lambda) \xa + \lambda \xa') \in \X$. Furthermore, since
\[
\alpha ((1 - \lambda) \xa + \lambda \xa')^* \leq (1 -\lambda) \alpha \xa^* + \lambda \alpha (\xa')^* \leq \alpha ((1 - \lambda) \xa + \lambda \xa'),
\]
we obtain $((1 - \lambda) \xa + \lambda \xa') \in \Xa$, which shows that $\Xa$ is convex for all $\alpha \in \zo$.

Furthermore, it holds that $\sup_{X \in \X} \expec \big[ X_\infty / \Xhat_\infty \big] \leq 1$ and, using Markov's inequality, we see that $\{ X_\infty / \Xhat_\infty \such X \in \X \}$ is bounded in $\prob$-measure. Since $\prob[\Xhat_\infty > 0] = 1$, the set $\set{X_\infty \such X \in \X}$ is bounded in $\prob$-measure; the same is then true for $\set{\xa_t \such X \in \X} \subseteq \set{X_t \such X \in \X} \subseteq \set{X_\infty \such X \in \X}$ for any value of $t \in [0, \infty]$.
\end{proof}

In the sequel, fix $\alpha \in (0,1)$. In view of Lemma \ref{lem: xa conv and bdd} above and Theorem 1.1(4) in \cite{MR2724418}, there exists a random variable $\chic_\infty$ in the closure in $\prob$-measure of $\set{X_\infty \such X \in \Xa}$ such that $\expec \bra{X_\infty / \chic_\infty} \leq 1$ holds for all $X \in \Xa$. 
Define the countable set $\Time = \set{k / 2^m \such k \in \Natural, \, m \in \Natural}$. A repeated application of Lemma A1.1 in \cite{MR1304434} combined with Lemma \ref{lem: xa conv and bdd} and a diagonalisation argument implies that one can find an $\Xa$-valued sequence $(X^n)_{\nin}$ such that $\chic_\infty = \lim_{n \to \infty} X^n_\infty$ and $\lim_{n \to \infty} X^n_t$ a.s. exists simultaneously for all $t \in \Time$. Define then $\chic_t = \lim_{n \to \infty} X^n_t$ for all $t \in \Time$. Since $T$ is a.s. finitely-valued and all processes are constant after $T$, it is straightforward that $\chic_\infty = \lim_{t \to \infty} \chic_t$ a.s.

Since $\expec [Y_t X^n_t \such \F_s] \leq Y_s X^n_s$ holds for all $\nin$, $Y \in \Xo$, $t \in \Time$ and $s \in \Time \cap [0, t]$, the conditional version of Fatou's lemma gives that $\expec [Y_t \chic_t \such \F_s] \leq Y_s \chic_s$ holds for all $Y \in \Xo$, $t \in \Time$ and $s \in \Time \cap [0, t]$. In particular, with $\Yhat \dfn 1 / \Xhat \in \Xo$, the process $(\Yhat_t \chic_t)_{t \in \Time}$ is a supermartingale in the corresponding stochastic basis with time-index $\Time$. Since $\prob[\inf_{s \in [0, t]} \Yhat_s > 0] = 1$ holds for all $t \in \Real_+$, the supermartingale convergence theorem implies that there exists a nonnegative \cadlag \ process $\chi$ such that\footnote{Note that $\chi$ is indeed the limit of $(X^n)_{\nin}$ in the ``Fatou'' sense. Fatou-convergence has proved to be extremely useful in the theory of Mathematical Finance; for example, see \cite{MR1469917}, \cite{MR1722287} and \cite{MR1883202}.} $\chi_s = \lim_{\Time \ni t \downarrow \downarrow s} \chic_t$ holds for all $s \in \Real_+$. (The notation ``$\lim_{\Time \ni t \downarrow \downarrow s}$'' denotes limit along times $t \in \Time$ that are \emph{strictly} greater than $s \in \Real_+$ and converge to $s$.) The fact that $\expec [Y_t \chic_t \such \F_s] \leq Y_s \chic_s $ holds for all $Y \in \Xo$, $t \in \Time$ and $s \in \Time \cap [0, t]$, right-continuity of the filtration $(\F_t)_{t \in \Real_+}$ and the conditional version of Fatou's lemma give that $\expec [Y_t \chi_t \such \F_s]  \leq  Y_s \chi_s$ holds for all $Y \in \Xo$, $t \in \Real_+$ and $s \in [0, t]$. Therefore, $\chi \in \Xoo$. Of course, $\chi_\infty = \chic_\infty = \lim_{t \to \infty} \chi_t$ a.s. holds. In view of Theorem \ref{thm: ODT}, it holds that $\chi = \xt (1 - A)$, where $\xt \in \X$ and $A$ is an adapted, nonnegative and nondecreasing \cadlag \ process with $0 \leq A \leq 1$. Furthermore, note that $\expec[X_\infty / \chi_\infty] \leq 1$ holds for all $X \in \xa$. 

Continuing, we shall show that $A \equiv 0$ and $\chi (= \xt) \in \Xa$. If $(X^n)_{\nin}$ is the $\Xa$-valued sequence such that $\chic_t = \lim_{n \to \infty} X^n_t$ holds a.s. simultaneously for all $t \in \Time$, we have that $X^n_t \geq \alpha X^n_s$ a.s. holds for all $t \in \Time$ and $s \in \Time \cap [0, t]$. By passing to the limit, and using the fact that $\Time$ is countable, we obtain that $\chic_t \geq \alpha \chic_s$ holds a.s. simultaneously for all $t \in \Time$ and $s \in \Time \cap [0, t]$. Therefore, $\chi_t \geq \alpha \chi_s$ holds a.s. simultaneously for all $t \in \Real_+$ and $s \in [0, t]$. Then,
\[
\xt_t = \frac{\chi_t}{1 - A_t} \geq  \frac{\chi_t}{1 - A_s} \geq \alpha  \frac{\chi_s}{1 - A_s} = \alpha \xt_s
\]
holds a.s. simultaneously for all $t \in \Real_+$ and $s \in [0, t]$. It follows that $\xt \in \Xa$. This implies, in particular, that $\expec[\xt_\infty / \chi_\infty] \leq 1$ has to hold. Since $\xt_\infty / \chi_\infty = 1 / (1 - A_\infty) \geq 1$, we obtain $\prob[A_\infty = 0] = 1$, i.e., $A \equiv 0$. Therefore, $\chi = \xt$ and $\expec[X_\infty / \xt_\infty] \leq 1$ holds for all $X \in \xa$, which concludes the proof of existence of a wealth process that possesses the \num \ property in $\Xa$ for investment over $[0, T]$.

\subsubsection{Uniqueness}

We proceed in establishing uniqueness of a process with the \num \ property in $\Xa$ for investment over the period $[0, \hor]$. We start by stating and proving a result that will be used again later.

\begin{lem} \label{lem: switching}
Let $Z \in \xa$, and let $\sigma$ be a stopping time such that $Z_\sigma = Z^*_\sigma$ a.s. holds on $\set{\sigma < \infty}$. Fix $X \in \Xa$ and $A \in \F_{\sigma}$ and define a new process\footnote{Note that, since we tacitly assume that $\alpha \in (0,1)$, $X_\sigma > 0$ a.s. holds on $\set{\sigma < \infty}$. Therefore, the process $\xi$ is well-defined.} $\xi \dfn Z \indic_{\dbraco{0, {\sigma}}} + \big( Z  \indic_{\Omega \setminus A} +  \pare{Z_{{\sigma}}  / X_{\sigma}  } X \indic_A \big) \indic_{\dbraco{{\sigma}, \infty}}$. Then, $\xi \in \Xa$.
\end{lem}

\begin{proof}
It is straightforward to check that $\xi \in \X$. To see that $\xi \in \Xa$, note that $\xi / \xi^* =  Z / Z^* \geq \alpha$ holds on $\dbraco{0, \sigma} \, \cup \, \pare{\dbraco{\sigma, \infty} \cap (\Omega \setminus A)}$, while, using the fact that $\xi^*_{\sigma} = Z^*_{\sigma} = Z_{\sigma}$ holds a.s.on $\set{\sigma < \infty}$,
\[
\frac{\xi}{\xi^*} = \frac{X}{\sup_{t \in [\sigma, \cdot]}  X_t} \geq  \frac{X}{X^*} \geq \alpha, \quad \text{holds on } \dbraco{\sigma, \infty} \cap A.
\]
The result immediately follows.
\end{proof}

\begin{rem}
As can be seen via the use of simple counter-examples, if one drops the assumption that $\sigma$ is a time of maximum of $Z$ in the statement of Lemma \ref{lem: switching}, the resulting process $\xi$ may fail to satisfy the drawdown constraints. This is in direct contrast with the non-constrained case $\alpha = 0$, where any stopping time $\sigma$ will result in $\xi$ being an element of $\X$. It is exactly this fact, a consequence of the non-myopic structure of the drawdown constraints, which results in portfolios with the \num \ property that depend on the investment horizon.
\end{rem}

\begin{lem} \label{lem: uniq_help}
Let $Z \in \Xa$ be such that $\errh(X|Z) \leq 0$ holds for all $X \in \xa$, and suppose that $\sigma$ is a stopping time such that $Z_\sigma = Z^*_\sigma$ a.s. holds on $\set{\sigma < \infty}$. Then,
\[
\expec \bra{ \frac{X_{T}}{Z_{T}} \ \Big| \ \F_{T \wedge \sigma}} \leq \frac{X_{T \wedge \sigma}}{Z_{T \wedge \sigma}} \text{ holds a.s. for all } X \in \Xa.
\]
\end{lem}

\begin{proof}
Fix $X \in \Xa$ and $A \in \F_{\sigma}$. Define the process $\xi \dfn Z \indic_{\dbraco{0, {\sigma}}} + \big( Z  \indic_{\Omega \setminus A} +  \pare{Z_{{\sigma}}  / X_{\sigma}  } X \indic_A \big) \indic_{\dbraco{{\sigma}, \infty}}$; by Lemma \ref{lem: switching}, $\xi \in \Xa$. Furthermore, it is straightforward to check that
\[
\frac{\xi_T}{Z_T} = \indic_{\Omega \setminus \pare{A \cap \set{\sigma \leq T}} } + \pare{ \frac{X_T}{Z_T} \frac{Z_{\sigma}}{X_{\sigma}} } \indic_{A \cap \set{\sigma \leq T}}.
\]
Therefore, the fact that $\errh(\xi|Z) \leq 0$ holds implies
\[
\expec \bra{ \frac{X_T}{Z_T} \frac{Z_{\sigma}}{X_{\sigma}} \indic_{A \cap \set{\sigma \leq T}}} \leq \prob[A \cap \set{\sigma \leq T} ].
\]
As the previous is true for all $A \in \F_{\sigma}$, we obtain that $\expec \bra{ X_{T} / Z_{T} \such \F_{\sigma}} \leq X_{\sigma} / Z_{\sigma}$ holds a.s. on $\set{\sigma \leq T}$ for all  $X \in \Xa$. Combined with the fact that $\expec \bra{ X_{T} / Z_{T} \such \F_{T}} = X_{T} / Z_{T}$ trivially holds a.s. on $\set{\sigma > T}$ for all  $X \in \Xa$, we obtain the result.
\end{proof}

We now proceed to the actual proof of uniqueness. Assume that both $\zt \in \Xa$ and $\xt \in \Xa$ have the \num \ property in $\Xa$ for investment over $[0, \hor]$. Since $\prob \bra{T < \infty} = 1$, Proposition \ref{prop: err} implies that $\prob \big[ \xt_T = \zt_T \big] = 1$. We shall show below that $\zt \leq \xt$ holds on $\dbra{0, T}$. Interchanging the roles of $\xt$ and $\zt$, it will also follow that $\xt \leq \zt$ holds on $\dbra{0, T}$, which will establish that $\xt = \zt$ holds on $\dbra{0, T}$ and will complete the proof of Theorem \ref{thm: exist}.

Since $\prob \big[ \xt_T = \zt_T \big] = 1$ and $\errh(X|\zt) \leq 0$ holds for all $X \in \Xa$, Lemma \ref{lem: uniq_help} above implies that $1 = \expec \big[ \xt_T / \zt_T \such \F_{T \wedge \sigma} \big] \leq \xt_{T \wedge \sigma} / \zt_{T \wedge \sigma}$ a.s. holds whenever $\sigma$ is a stopping time such that $\zt_\sigma = \zt^*_{\sigma}$ a.s. holds on $\set{\sigma < \infty}$. The fact that $\zt_{T \wedge \sigma} \leq \xt_{T \wedge \sigma}$ a.s. holds whenever $\sigma$ is a stopping time such that $\zt_\sigma = \zt^*_{\sigma}$ a.s. holds on $\set{\sigma < \infty}$ implies in a straightforward way that $\zt^* \leq \xt^*$ holds on $\dbra{0, T}$.

We now claim that $\prob \big[ \zt_T = \xt_T \big] = 1$ combined with $\zt^* \leq \xt^*$ holding on $\dbra{0, T}$ imply that $\zt \leq \xt$ on $\dbra{0, T}$, which will complete the proof. To see the last claim, for $\epsilon > 0$ define the stopping time
\[
T_\epsilon \dfn \inf \set{t \in \Real_+ \such \zt_t > (1 + \epsilon) \xt_t}.
\]
We shall show that $\prob \bra{T_\epsilon < T} = 0$; as this will hold for all $\epsilon > 0$, it will follow that $\zt \leq \xt$ holds on $\dbra{0, T}$.
Define a new process $\xt^\epsilon$ via
\[
\xt^\epsilon = \zt \indic_{\dbraco{0, T_\epsilon}} + \pare{ \frac{\zt_{T^\epsilon}}{\xt_{T^\epsilon}} }  \xt \indic_{\dbraco{T^\epsilon, \infty}} = \zt \indic_{\dbraco{0, T_\epsilon}} + \pare{ 1 + \epsilon } \xt \indic_{\dbraco{T^\epsilon, \infty}}.
\]
We first show that $\xt^\epsilon \in \Xa$. The fact that $\xt \in \X$ is obvious. Note also that $\xt^\epsilon \geq \alpha (\xt^\epsilon)^*$ clearly holds on $\dbraco{0, T^\epsilon}$, since $\zt \in \xa$. On the other hand,
\[
(\xt^\epsilon)^*_t = (\zt_{T^\epsilon})^* \vee \sup_{s \in [T^\epsilon, t]} \pare{(1 + \epsilon) \xt_t} \leq (1 + \epsilon) \xt^*_t \text{ holds for } t \geq T^\epsilon,
\]
the latter inequality holding in view of the fact that $\zt^* \leq \xt^*$. Therefore, for $t \geq T^\epsilon$ it holds that $\xt^\epsilon_t = (1 + \epsilon) \xt_t \geq (1 + \epsilon) \alpha \xt^*_t \geq \alpha (\xt^\epsilon)^*_t$. It follows that $\xt^\epsilon \geq \alpha (\xt^\epsilon)^*$ also holds on $\dbraco{T^\epsilon, \infty}$, which shows that $\xt^\epsilon \in \Xa$. Note that
\[
\xt^\epsilon_T = \zt_T \indic_{\set{T < T^\epsilon}} + (1 + \epsilon) \xt_T \indic_{\set{T^\epsilon \leq T}} =  \xt_T \indic_{\set{T < T^\epsilon}} + (1 + \epsilon) \xt_T \indic_{\set{T^\epsilon \leq T}},
\]
which implies that $\xt^\epsilon_T / \xt_T = 1 + \epsilon \indic_{\set{T^\epsilon \leq T}}$ and, as a consequence, $\errh (\xt^\epsilon|\xt) = \epsilon \prob \bra{T^\epsilon \leq T}$. In case $\prob \bra{T^\epsilon \leq T} > 0$, it would follow that $\xt$ fails to have the \num \ property in $\Xa$ for investment in $[0, T]$. Therefore, $\prob \bra{T^\epsilon \leq T} = 0$, which implies that $\zt \leq \xt$ holds on $\dbra{0, T}$, as already mentioned. The proof of Theorem \ref{thm: exist} is complete.

\subsection{Proof of Theorem \ref{thm: num}} \label{subsec: proof of thm:num}

The main tool towards proving assertion (1) of Theorem \ref{thm: num} is the following auxiliary result.

\begin{lem} \label{lem: conv of max}
For any $X \in \X$, $\lim_{t \to \infty} ( X^*_t / \Xhat^*_t )$ a.s. exists. Moreover, it a.s. holds that
\[
\lim_{t \to \infty} \pare{ \frac{X^*_t}{\Xhat^*_t} } = \lim_{t \to \infty} \pare{ \frac{X_t}{\Xhat_t} }.
\]
\end{lem}

\begin{proof}
For $t \in \Real_+$, define the $[0,t]$-valued random time $\rhoh_t \dfn \sup \{ s \in [0, t] \such \Xhat_s = \Xhat^*_s \}$; then, $\Xhat^*_t = \Xhat_{\rhoh_t }$. Note that $\prob \bra{\uparrow \lim_{t \to \infty} \rhoh_t = \infty} = 1$ holds in view of Assumption \ref{ass: basic}. It follows that, for any $X \in \X$, it a.s. holds that
\begin{equation} \label{eq: conv of max inf}
\liminf_{t \to \infty} \pare{\frac{X^*_t}{\Xhat^*_t}} = \liminf_{t \to \infty} \pare{\frac{X^*_t}{\Xhat^*_{\rhoh_t}}} \geq \liminf_{t \to \infty} \pare{\frac{X_{\rhoh_t}}{\Xhat_{\rhoh_t}}} =  \lim_{t \to \infty} \pare{ \frac{X_t}{\Xhat_t} }.
\end{equation}

In what follows, fix $X \in \X$. For $t \in \Real_+$ define $\rho_t \dfn \sup \set{s \in [0, t] \such X_s = X^*_s}$, which is a $[0,t]$-valued random time. For each $t \in \Real_+$, $X^*_t = X_{\rho_t}$. Note that the set-inclusions $\set{\uparrow \lim_{t \to \infty} \rho_t < \infty} \subseteq \set{\sup_{t \in \Real_+} X_t < \infty} \subseteq \{ \rri (X | \Xhat) = -1 \}$ are valid a.s., the last in view of Assumption \ref{ass: basic}. Therefore,
\begin{equation} \label{eq: conv of max sup 1}
\lim_{t \to \infty} \pare{\frac{X_t}{\Xhat_t}} = \lim_{t \to \infty} \pare{\frac{X^*_t}{\Xhat^*_t}} = 0 \text{ holds on } \set{\lim_{t \to \infty} \rho_t < \infty}.
\end{equation}
Furthermore,
\begin{equation} \label{eq: conv of max sup 2}
\limsup_{t \to \infty} \pare{\frac{X^*_t}{\Xhat^*_t}} = \limsup_{t \to \infty} \pare{\frac{X^*_{\rho_t}}{\Xhat^*_t}} \leq \limsup_{t \to \infty} \pare{\frac{X_{\rho_t}}{\Xhat_{\rho_t}}} =  \lim_{t \to \infty} \pare{ \frac{X_t}{\Xhat_t} } \text{ holds on } \set{\lim_{t \to \infty} \rho_t = \infty}.
\end{equation}
The claim now readily follows from \eqref{eq: conv of max inf}, \eqref{eq: conv of max sup 1}, and \eqref{eq: conv of max sup 2}.
\end{proof}

\begin{proof}[Proof of Theorem \ref{thm: num}, statement (1)]
In the sequel, fix $X \in \X$ and assume that $\alpha \in (0,1)$. Results for the case $\alpha = 0$ are well-understood and not discussed.

To ease notation, let $D \dfn X / X^*$ and $\Dh \dfn \Xhat / \Xhat^*$. The process $D$ is $[0,1]$-valued and $\Dh$ is $(0,1]$-valued. Observe that
\[
\frac{\xa}{\xha} = \frac{\alpha (X^*)^{1 - \alpha} + \alpha (X^*)^{- \alpha} X}{\alpha (\Xhat^*)^{1 - \alpha} + \alpha (\Xhat^*)^{- \alpha} \Xhat} = \pare{\frac{X^*}{\Xhat^*}}^{1 - \alpha} \pare{\frac{\alpha + (1 - \alpha) D}{\alpha + (1 - \alpha) \Dh}}.
\]
In view of Lemma \ref{lem: conv of max}, $\lim_{t \to \infty} ( X^*_t / \Xhat^*_t)^{1 - \alpha} = (1 + \rri(X|\Xhat))^{1 - \alpha}$ holds. Firstly, the fact that
\[
\frac{\alpha + (1 - \alpha) D}{\alpha + (1 - \alpha) \Dh} \leq \frac{1}{\alpha}
\]
implies that $\xa / \xha \leq (1/\alpha) (X^* / \Xhat^*)^{1 - \alpha}$, which readily gives \eqref{eq: asympt num rel} on $\{ \rri(X|\Xhat) = -1 \}$. Furthermore, the facts that $0 \leq D \leq 1$, $0 < \Dh \leq 1$ and $\lim_{t \to \infty} (D_t / \Dh_t) = 1$, the latter holding a.s. on $\{ \rri(X|\Xhat) > -1 \}$ in view of Lemma \ref{lem: conv of max}, imply that
\[
\limsup_{t \to \infty} \abs{\frac{\alpha + (1 - \alpha) D_t}{\alpha + (1 - \alpha) \Dh_t} - 1} \leq \frac{1 - \alpha}{\alpha} \limsup_{t \to \infty} |D_t - \Dh_t| = 0 \text{ holds on } \set{\rri(X|\Xhat) > -1}.
\]
Therefore, $\lim_{t \to \infty} ( \xa_t / \xha_t ) = (1 + \rri(X|\Xhat))^{1 - \alpha}$ also holds on the event $\{\rri(X|\Xhat) > -1 \}$, which completes the proof of statement (1) of Theorem \ref{thm: num}.
\end{proof}

\begin{proof}[Proof of Theorem \ref{thm: num}, statement (2)]
Let $\tau$ be a time of maximum of $\Xhat$. Recall the definition of the stopping times $(\tau_\ell)_{\ell \in \Real_+}$ from \eqref{eq: level_crossing}. In view of statement (1) of Theorem \ref{thm: num},
\begin{equation} \label{eq: lim through tau_ell}
\rrt(\xa|\xha) = \lim_{\ell \to \infty} \pare{ \frac{\xa_{\tau \wedge \tau_\ell}}{\xha_{\tau \wedge \tau_\ell}} } - 1
\end{equation}
a.s. holds. Now, observe that $\tau \wedge \tau_\ell$ is a time of maximum of $\Xhat$ for each $\ell \in \Real_+$; therefore, $\xha_{\tau \wedge \tau_\ell} = (\Xhat_{\tau \wedge \tau_\ell})^{1 - \alpha} = (\Xhat^*_{\tau \wedge \tau_\ell})^{1 - \alpha}$.  It then follows that
\begin{equation} \label{eq: ratio}
\frac{\xa_{\tau \wedge \tau_\ell}}{\xha_{\tau \wedge \tau_\ell}} = \alpha \pare{ \frac{X^*_{\tau \wedge \tau_\ell} }{\Xhat^*_{\tau \wedge \tau_\ell}} }^{1 - \alpha} + (1 - \alpha) \pare{ \frac{X_{\tau \wedge \tau_\ell} }{\Xhat_{\tau \wedge \tau_\ell}} } \pare{ \frac{X^*_{\tau \wedge \tau_\ell} }{\Xhat^*_{\tau \wedge \tau_\ell}} }^{- \alpha}.
\end{equation}
Define $\chi \dfn X / \Xhat$ and, in the obvious way, $\chi^* \dfn \sup_{t \in [0, \cdot]} (X_t / \Xhat_t)$. For $y \in \Real_+$, the function $[y, \infty) \ni z \mapsto \alpha z^{1-\alpha} + (1 - \alpha) y z^{-\alpha}$ is nondecreasing, which can be shown upon simple differentiation. With $y = \chi_{\tau \wedge \tau_\ell}$, $z_1 = X^*_{\tau \wedge \tau_\ell} / \Xhat^*_{\tau \wedge \tau_\ell} = X^*_{\tau \wedge \tau_\ell} / \Xhat_{\tau \wedge \tau_\ell} \geq y$ and $z_2 = \chi^*_{\tau \wedge \tau_\ell} \geq X^*_{\tau \wedge \tau_\ell} / \Xhat^*_{\tau \wedge \tau_\ell} = z_1$, \eqref{eq: ratio} then implies that
\[
\frac{\xa_{\tau \wedge \tau_\ell}}{\xha_{\tau \wedge \tau_\ell}} \leq \alpha \pare{ \chi^*_{\tau \wedge \tau_\ell} }^{1 - \alpha} + (1 - \alpha) \chi_{\tau \wedge \tau_\ell} \pare{ \chi^*_{\tau \wedge \tau_\ell} }^{- \alpha}.
\]
Define the process $\phi \dfn \alpha \pare{ \chi^* }^{1 - \alpha} + (1 - \alpha) \chi \pare{ \chi^* }^{- \alpha}$; then, by the last estimate and \eqref{eq: lim through tau_ell},
\[
\rrt(\xa|\xha) \leq \liminf_{\ell \to \infty} \pare{ \phi_{\tau \wedge \tau_\ell} } - 1.
\]
Since $\int_{0}^\infty \indic_{\set{\chi_t < \chi^*_t}} \ud \chi^*_t = 0$ a.s. holds, a straightforward use of It\^o's formula gives
\[
\phi = 1 + \int_0^\cdot (1 - \alpha) \pare{ \chi_t^* }^{- \alpha} \ud \chi_t;
\]
since $\chi$ is a local martingale, $\phi$ is a local martingale as well. Since $\phi$ is nonnegative, it is a supermartingale with $\phi_0 = 1$, which implies that $\expec \bra{\phi_{\tau \wedge \tau_\ell}} \leq 1$ holds for all $\ell \in \Real_+$. It follows that
\[
\errt(\xa|\xha) = \expec \bra{\rrt(\xa|\xha)} \leq \expec \bra{\liminf_{\ell \to \infty} \pare{\phi_{\tau \wedge \tau_\ell}} } - 1 \leq \liminf_{\ell \to \infty} \pare{ \expec \bra{\phi_{\tau \wedge \tau_\ell}} } - 1 \leq 0,
\]

Now, let $\sigma$ be a time of maximum of $\Xhat$ with $\sigma \leq \tau$. Fix $X \in \X$ and $A \in \F_{\sigma}$; by Lemma \ref{lem: switching}, the process $\xi \dfn \xha \indic_{\dbraco{0, {\sigma}}} + \big( \xha  \indic_{\Omega \setminus A} + \xha_{{\sigma}} \pare{\xa  / \xa_{\sigma}  } \indic_A \big) \indic_{\dbraco{{\sigma}, \infty}}$ is an element of $\Xa$. Furthermore, it is straightforward to check that
\[
\rrt(\xi|\xha) = \pare{\frac{1 + \rrt(\xa|\xha)}{1 + \rr_{\sigma}(\xa|\xha)} - 1} \indic_{A \cap \set{\sigma < \infty}}.
\]
Since $\errt(\xi|\xha) \leq 0$ has to hold by the result previously established, we obtain
\[
\expec \bra{\frac{1 + \rrt(\xa|\xha)}{1 + \rr_{\sigma}(\xa|\xha)} \indic_{A \cap \set{\sigma < \infty}}} \leq  \prob[A \cap \set{\sigma < \infty}].
\]
Since the previous holds for all $A \in \F_\sigma$, we obtain that $\expec \big[ \rrt(\xa|\xha) \such \F_\sigma \big] \leq \rr_{\sigma}(\xa|\xha)$ holds on $\set{\sigma < \infty}$. On $\set{\sigma = \infty}$, we have $\sigma=\tau$ and $\expec \big[ \rrt(\xa|\xha) \such \F_\sigma \big] = \rri(\xa|\xha) = \rr_{\sigma}(\xa|\xha)$. Therefore, $\expec \big[ \rrt(\xa|\xha) \such \F_\sigma \big] \leq \rr_{\sigma}(\xa|\xha)$ holds.
\end{proof}

\subsection{Proof of Theorem \ref{thm: asympt_growth}}\label{subsec: proof of thm: asympt_growth}
The fact that $\lim_{t \to \infty} \pare{\log\big( \xhat_t \big) / G_t} = 1$ holds on the event $\set{G_\infty = \infty}$ was established in the proof of Theorem \ref{thm: asss}. Again, in view of Theorem \ref{thm: asss}, condition (A2) of Assumption \ref{ass: basic} is equivalent to $\prob \bra{G_\infty = \infty} = 1$; therefore, a.s.,
\[
\lim_{t \to \infty} \pare{\frac{1}{G_t} \log(\Xhat_t)} = 1.
\]
Observe that by concavity of the function $\Real_+ \ni x \mapsto x^{1-\alpha}$, $\xha\geq \Xhat^{1-\alpha}$ holds. Combining this with the above yields, a.s.,
\[
\liminf_{t \to \infty} \pare{\frac{1}{G_t} \log(\xha_t)} \geq 1 - \alpha.
\]
On the other hand, since $G$ is nondecreasing and $\xha$ achieves maximum values at the times $(\tau_\ell)_{\ell \in \Real_+}$ of \eqref{eq: level_crossing}, it holds a.s.\ that
\begin{align*}
\limsup_{t \to \infty} \pare{\frac{1}{G_t} \log(\xha_t)} &= \limsup_{\ell \to \infty} \pare{\frac{1}{G_{\tau_\ell}} \log(\xha_{\tau_\ell})} \\
&= (1 - \alpha) \limsup_{\ell \to \infty} \pare{\frac{ \ell}{G_{\tau_\ell}}} \\
&= (1 - \alpha) \limsup_{\ell \to \infty} \pare{\frac{1}{G_{\tau_\ell}} \log(\Xhat_{\tau_\ell})} = 1 - \alpha.
\end{align*}
It follows that, a.s.,
\[
\lim_{t \to \infty} \pare{\frac{1}{G_t} \log(\xha_t)} = 1 - \alpha.
\]

Fix $Z \in \Xa$. The full result of Theorem \ref{thm: asympt_growth} now follows immediately upon noticing that, a.s.,
\[
\limsup_{t \to \infty} \pare{\frac{1}{G_t} \log \pare{\frac{Z_t}{\xha_t} }} \leq 0,
\]
which is valid in view of the facts that $\prob[G_\infty = \infty]= 1$ and $\prob \big[ \rri(Z|\xha) < \infty \big]  = 1$, the latter following from the inequality $\erri(Z|\xha) \leq 0$, which was established in Theorem \ref{thm: num}.

\subsection{Proof of Proposition \ref{prop: re_asympt_num}}\label{subsec: proof of prop: re_asympt_num}
Upon passing to a subsequence of $(T_n)_{\nin}$ if necessary, we may assume without loss of generality that $\prob \bra{\limn T_n = \infty} = 1$. Then, by Theorem \ref{thm: num}, $\lim_{t \to \infty} \big( \xha_t / \xa_t \big)$ exists a.s. in $(0, \infty]$ and a use of Fatou's lemma gives
\[
\expec \bra{\lim_{t \to \infty} \pare{ \frac{\xha_t}{\xa_t} } } = \expec \bra{\liminf_{n \to \infty} \pare{ \frac{\xha_{T_n}}{\xa_{T_n}} } } \leq \liminf_{n \to \infty} \pare{\expec \bra{ \frac{\xha_{T_n}}{\xa_{T_n}} } } = 1 + \liminf_{n \to \infty} \err_{T_n} (\xha|\xa) \leq 1.
\]
 Since we have both $\expec \big[ \lim_{t \to \infty} \big( \xha_t / \xa_t \big) \big] \leq 1$ and $\expec \big[ \lim_{t \to \infty} \big( \xa_t / \xha_t \big) \big] \leq 1$ holding, Jensen's inequality implies that $\lim_{t \to \infty} \big( \xa_t / \xha_t \big) = 1$ a.s. holds. By Theorem \ref{thm: num}, $\lim_{t \to \infty} \big( X_t / \Xhat_t \big) = 1$ a.s. holds. This fact, combined with the conditional form of Fatou's lemma and the supermartingale property of $X /\Xhat$ gives $X_t / \Xhat_t \geq 1$ a.s.\ for each $t \in \Real_+$. Combined with $\expec \big[ X_t / \Xhat_t \big] \leq 1$, this gives $\Xhat_t = X_t$ a.s.\ for all $t \in \Real_+$. The path-continuity of the process $X/\Xhat$ implies that $X = \Xhat$, i.e., that $\xa = \xha$.

\subsection{Proof of Theorem \ref{thm: turnpike}} \label{subsec: proof of thm:turnpike}

In the setting of Definition \ref{dfn: emery}, consider a sequence $(\xi^n)_{\nin}$ of semimartingales and another semimartingale $\xi$. It is straightforward to check that $\Sln \xi^n = \xi$ holds if and only if there exists a nondecreasing sequence $(\tau_k)_{\kin}$ of finitely-valued stopping times with $\prob \bra{\lim_{k \to \infty} \tau_k = \infty} = 1$ such that $\Stkn \xi^n = \xi$ holds for all $\kin$. For the proof of Theorem \ref{thm: turnpike}, we shall use the previous observation along the sequence $(\tau_\ell)_{\ell \in \Natural}$ of finitely-valued stopping times defined in \eqref{eq: level_crossing}. Therefore, in the course of the proof, we keep $\ell \in \Real_+$ fixed and will show that $\Stln \xta^n = \xha$.

As a first step, we shall show that $\plimn \xta^n_{\tau_\ell} = \xha_{\tau_\ell}$, where ``$\plim$'' denotes limit in probability. For each $\nin$, consider the process $\xi^n \dfn \xha \indic_{\dbraco{0, {\tau_\ell}}} + \xha_{{\tau_\ell}} \pare{\xta^n  / \xta^n_{\tau_\ell}  } \indic_{\dbraco{{\tau_\ell}, \infty}}$. By Lemma \ref{lem: switching}, $\xi^n \in \Xa$ for all $\nin$. Furthermore, note that
\[
\rr_{T_n} (\xi^n | \xta) = \rr_{\tau_\ell} (\xha | \xta^n) \indic_{\set{\tau_\ell < T_n}} + \rr_{T_n} (\xha | \xta^n) \indic_{\set{T_n \leq \tau_\ell}} = \rr_{T_n \wedge \tau_\ell} (\xha | \xta^n).
\]
Using the previous relationship, the assumptions of Theorem \ref{thm: turnpike} give $\err_{T_n \wedge \tau_\ell} (\xha | \xta^n) \leq 0$ for all $\nin$. Furthermore, by Theorem \ref{thm: num}, $\err_{\tau_\ell} (\xta^n | \xha) \leq 0$ holds for all $\nin$. Therefore, $\expec\big[ \rr_{T_n \wedge \tau_\ell} (\xha | \xta^n) + \rr_{\tau_\ell} (\xta^n | \xha) \big] \leq 0$ holds for all $\nin$. Observe that the equality $\rr_{T_n \wedge \tau_\ell} (\xha | \xta^n) + \rr_{\tau_\ell} (\xta^n | \xha ) = (\xta^n_{\tau_\ell} - \xha_{\tau_\ell})^2 / (\xha_{\tau_\ell} \xta^n_{\tau_\ell})$ holds on $\set{\tau_\ell < T_n}$, and that the inequality $\rr_{T_n \wedge \tau_\ell} (\xha | \xta^n) + \rr_{\tau_\ell} (\xta^n | \xha) \geq -2$ is always true; therefore,
\[
\rr_{T_n \wedge \tau_\ell} (\xha | \xta^n) + \rr_{\tau_\ell} (\xta^n | \xha) \geq \frac{(\xta^n_{\tau_\ell} - \xha_{\tau_\ell})^2}{\xha_{\tau_\ell} \xta^n_{\tau_\ell}} \indic_{\set{\tau_\ell < T_n}} - 2 \indic_{\set{T_n \leq \tau_\ell}}.
\]
Since $\expec\big[ \rr_{T_n \wedge \tau_\ell} (\xha | \xta^n) + \rr_{\tau_\ell} (\xta^n | \xha) \big] \leq 0$ holds for all $\nin$ and $\limn \prob \bra{T_n \leq \tau_\ell} = 0$ holds in view of Theorem \ref{thm: asss}, we obtain that
\[
\limn \expec \bra{\frac{(\xta^n_{\tau_\ell} - \xha_{\tau_\ell})^2}{\xha_{\tau_\ell} \xta^n_{\tau_\ell}} \indic_{\set{\tau_\ell < T_n}}} = 0.
\]
Using again the fact that $\limn \prob \bra{\tau_\ell < T_n} = 1$, we obtain that $\plimn \xta^n_{\tau_\ell} = \xha_{\tau_\ell}$.

Given $\plimn \xta^n_{\tau_\ell} = \xha_{\tau_\ell}$, we now proceed in showing that $\plimn (\xt^n_{\tau_\ell} / \xhat_{\tau_\ell}) = 1$. We use some arguments similar to the first part of the proof of statement (2) of Theorem \ref{thm: num}, where the reader is referred to for certain details that are omitted here. Define $\chi^n \dfn \xt^n / \Xhat$ and $(\chi^n)^* \dfn \sup_{t \in [0, \cdot]} (\xt^n_t / \Xhat_t)$. It then follows that
\begin{equation} \label{eq: sandw_turnpike}
\frac{\xta^n_{\tau_\ell}}{\xha_{\tau_\ell}} \leq \alpha \pare{ (\chi^n)^*_{\tau_\ell} }^{1 - \alpha} + (1 - \alpha) \chi^n_{\tau_\ell} \pare{ (\chi^n)^*_{\tau_\ell} }^{- \alpha} =: \phi^n_{\tau_\ell},
\end{equation}
where the process $\phi^n \dfn \alpha \pare{ (\chi^n)^* }^{1 - \alpha} + (1 - \alpha) \chi^n \pare{ (\chi^n)^* }^{- \alpha}$ is a nonnegative local martingale for each $\nin$. We claim that $\plimn \phi^n_{\tau_\ell} = 1$. To see this, first observe that $\pliminfn \phi^n_{\tau_\ell} \geq 1$ holds, in the sense that $\liminfn \prob \bra{\phi^n_{\tau_\ell} > 1 - \epsilon} \geq \liminfn \prob \bra{\xta^n_{\tau_\ell} /  \xha^n_{\tau_\ell} > 1 - \epsilon} = 1$ holds for all $\epsilon \in (0,1)$. Then, given that $\pliminfn \phi^n_{\tau_\ell} \geq 1$, if $\limsupn \prob \bra{\phi^n_{\tau_\ell} > 1 + \epsilon} > 0$ was true, one would conclude that $\limsupn \expec \bra{\phi^n_{\tau_\ell}} > 1$, which contradicts the fact that $\expec \bra{\phi^n_{\tau_\ell}} \leq \phi^n_{0} = 1$ holds for all $\nin$. Therefore, $\limsupn \prob \bra{\phi^n_{\tau_\ell} > 1 + \epsilon} = 0$ holds for all $\epsilon \in (0,1)$, which combined with $\pliminfn \phi^n_{\tau_\ell} \geq 1$ gives $\plimn \phi^n_{\tau_\ell} = 1$. To recapitulate, the setting is the following: $(\phi^n)_{\nin}$ is a sequence of nonnegative local martingales with $\phi^n_0 = 1$, and  $\plimn \phi^n_{\tau_\ell} = 1$ holds. In that case, Lemma 2.11 in \cite{Kar11} implies that $\plimn (\phi^n)^*_{\tau_\ell} = 1$ holds as well. Note that $(\phi^n)^* = ((\chi^n)^*)^{1 - \alpha}$, so that $\plimn (\chi^n)^*_{\tau_\ell} = 1$ holds as well. Then, the bounds in \eqref{eq: sandw_turnpike} imply that $\plimn \chi^n_{\tau_\ell} = 1$.

Once again, we are in the following setting: $(\chi^n)_{\nin}$ is a sequence of nonnegative local martingales with $\chi^n_0 = 1$, and  $\plimn \chi^n_{\tau_\ell} = 1$ holds.
An application of Proposition 2.7 and Lemma 2.12 in \cite{Kar11} gives that $\Stln \chi^n = 1$, which also implies that $\Stln \xt^n = \xhat$ by Proposition 2.10 in \cite{Kar11}. This implies that $\limn \prob \big[\sup_{t \in [0, \tau_\ell]} |(\xt^n)^*_t - \xhat^*_t| > \epsilon \big] = 0$ also holds for all $\epsilon > 0$ by Remark \ref{rem: ucp_conv}. Therefore, by \eqref{eq: xa_SDE} and Lemma 2.9 in \cite{Kar11}, we obtain that $\Stln \xta^n = \xha$, which completes the proof of Theorem \ref{thm: turnpike}.

\section{A Cautionary Note Regarding Theorem \ref{thm: turnpike}} \label{sec: careful_turnpike}

In this Section, we elaborate on the point that is made in Remark \ref{rem: careful_turnpike} via use of an example. In the discussion that follows, fix $\alpha \in (0,1)$. The model is the general one described in Subsection \ref{subsec: model}, and Assumption \ref{ass: basic} is always in force.

Let $T_{1/2} = 0$ and, using induction, for $\nin$ define
\[
T_{n} \dfn \inf \{ t \in (T_{n - 1/2}, \infty) \such \Xhat_t = \alpha \Xhat^*_t \}, \quad T_{n + 1/2} \dfn \inf \{ t \in (T_n, \infty) \such \Xhat_t = \Xhat^*_{T_n} \}.
\]
(In the setting of Example \ref{exa: time-horizon_matters}, $T$ there is exactly $T_1$ defined above.) Note the following: $T_{n - 1/2}$ is a time of maximum of $\Xhat$ for all $\nin$, $(T_{k/2})_{\kin}$ is an increasing sequence, and $\prob \bra{\limn T_n = \infty} = 1$ holds. Under Assumption \ref{ass: basic}, Lemma \ref{lem: fin_dd_time} implies that $\prob \bra{T_n < \infty} = 1$ for all $\nin$.

For each $\nin$, one can explicitly describe the wealth process $\xta^n$ that has the \num \ property in the class $\Xa$ for investment over the interval $[0, T_n]$. In words, $\xta^n$ will follow $\xha$ until time $T_{n - 1/2}$, then switch to investing like the \num \ portfolio $\Xhat$ up to time $T_n$ and, since at time $T_n$ one hits the hard drawdown constraint, $\xta^n$ will remain constant from $T_n$ onwards. In mathematical terms, define
\begin{align*}
\xta^n :&= \xha \indic_{\dbraco{0, T_{n - 1/2}}} + \pare{\frac{\xha_{T_{n - 1/2}}}{\xhat_{T_{n - 1/2}}} } \xhat \indic_{\dbraco{T_{n - 1/2}, T_n}} +  \pare{\frac{\xha_{T_{n - 1/2}}}{\xhat_{T_{n - 1/2}}} } \xhat_{T_n} \indic_{\dbraco{T_n, \infty}}\\
&= \xha \indic_{\dbraco{0, T_{n - 1/2}}} + \pare{\xhat_{T_{n - 1/2}} }^{- \alpha} \xhat \indic_{\dbraco{T_{n - 1/2}, T_n}} +  \pare{\xhat_{T_{n - 1/2}} }^{- \alpha} \alpha \xhat^*_{T_n} \indic_{\dbraco{T_n, \infty}},
\end{align*}
where for the equality in the second line the facts that $\xha_{T_{n - 1/2}} = (\xhat_{T_{n - 1/2}})^{1 - \alpha}$ and $\xhat_{T_n} = \alpha \xhat^*_{T_n}$ were used. It is straightforward to check that $\xta \in \Xa$, in view of the definition of the stopping times $(T_{k/2})_{\kin}$. Pick any $X \in \X$. The global (in time) \num \ property of $\Xhat$ in $\X$ will give
\[
\expec \bra{ \frac{\xa_{T_{n}}}{ \xta^n_{T_{n}} } - 1 \, \bigg| \, \F_{T_{n - 1/2}}} \leq \frac{\xa_{T_{n - 1/2}}}{ \xta^n_{T_{n - 1/2}}} - 1 = \frac{\xa_{T_{n - 1/2}}}{ \xha^n_{T_{n - 1/2}}} - 1.
\]
Upon taking expectation on both sides of the previous inequality, we obtain $\err_{T_n} (\xa | \xta^n) \leq \err_{T_{n-1/2}} (\xa | \xha^n) \leq 0$, the last inequality holding in view of statement (2) of Theorem \ref{thm: num}, given that $T_{n - 1/2}$ is a time of maximum of $\Xhat$. We have shown that $\xta^n$ indeed has the \num \ property in the class $\Xa$ for investment over the interval $[0, T_n]$.

Note that $\xta^n = \xha$ identically holds in the stochastic interval $\dbra{0, T_{n-1/2}}$ for each $\nin$; therefore, the conclusion of Theorem \ref{thm: turnpike} in this case is valid in a quite strong sense. However, the behaviour of $\xta^n$ and $\xha$ in the stochastic interval $\dbra{T_{n-1/2}, T_n}$ is different and results in quite diverse outcomes at time $T_n$, as we shall now show. At time $T_n$ one has
\[
\xha_{T_n} = \alpha (\Xhat^*_{T_n})^{1 - \alpha} + (1 - \alpha)(\Xhat^*_{T_n})^{- \alpha} \Xhat_{T_n} = \alpha (2 - \alpha) (\Xhat^*_{T_n})^{1 - \alpha},
\]
where the fact that $\xhat_{T_n} = \alpha \xhat^*_{T_n}$ was again used. Furthermore, $\xta^n_{T_n} =  (\xhat_{T_{n - 1/2}} )^{- \alpha} \alpha \xhat^*_{T_n}$. It then follows that
\[
\frac{\xta^n_{T_n}}{\xha_{T_n}} = \frac{(\xhat_{T_{n - 1/2}} )^{- \alpha} \alpha \xhat^*_{T_n}}{\alpha (2 - \alpha) (\Xhat^*_{T_n})^{1 - \alpha}} = \frac{1}{2 - \alpha} \pare{\frac{\Xhat^*_{T_n}}{\xhat_{T_{n - 1/2}}} }^{\alpha}=:\zeta_n.
\]
In view of Assumption \ref{ass: basic} and the result of Dambis, Dubins and Schwarz---see Theorem 3.4.6 in \cite{MR1121940}---the law of the random variable $\zeta_n$ is the same for all $\nin$. In fact, universal distributional properties of the maximum of a non-negative local martingale stopped at first hitting time---see Proposition 4.3 in \cite{BKO}---imply that $\zeta_n = (2-\alpha)^{-1} (\alpha+\pare{1-\alpha} \pare{1 / \eta_n})^\alpha$, where $\eta_n$ has the uniform law on $(0,1)$. In particular, $\prob [\zeta_n < (2 - \alpha)^{-1} + \epsilon] > 0$ and $\prob [\zeta_n > (2 - \alpha)^{-1} + \epsilon^{-1}] > 0$ holds for all $\epsilon \in (0,1)$. Furthermore, $\zeta_n$ is $\F_{T_n}$-measurable and independent of $\F_{T_{n - 1/2}} \supseteq \F_{T_{n-1}}$ for each $\nin$, which implies that $(\zeta_n)_{\nin}$ is a sequence of independent and identically distributed random variables. By an application of the second Borel-Cantelli lemma, it follows that
\[
\frac{1}{2 - \alpha} = \liminf_{n \to \infty} \pare{\frac{\xta^n_{T_n}}{\xha_{T_n}}} < \limsup_{n \to \infty} \pare{\frac{\xta^n_{T_n}}{\xha_{T_n}}} = \infty,
\]
demonstrating the claim made at Remark \ref{rem: careful_turnpike}.

\bibliographystyle{agsm}
\bibliography{num_dd5}
\end{document}